\pdfoutput=1
\RequirePackage{ifpdf}
\ifpdf 
\documentclass[pdftex]{sigma}
\else
\documentclass{sigma}
\fi

\usepackage{tikz}
\usepackage[all]{xy}
\numberwithin{equation}{section}

\newtheorem{Theorem}{Theorem}[section]
\newtheorem*{Theorem*}{Theorem}
\newtheorem{Corollary}[Theorem]{Corollary}
\newtheorem{Lemma}[Theorem]{Lemma}
\newtheorem{Proposition}[Theorem]{Proposition}
 { \theoremstyle{definition}

\newtheorem{Example}[Theorem]{Example}
\newtheorem{Remark}[Theorem]{Remark} }

\newcommand{\al}{\alpha}

\newcommand{\La}{\Lambda}
\newcommand{\la}{\lambda}

\newcommand{\pa}{\partial}

\begin{document}

\allowdisplaybreaks

\newcommand{\arXivNumber}{2408.09450}

\renewcommand{\PaperNumber}{110}

\FirstPageHeading

\ShortArticleName{The Modified Toda Hierarchy}

\ArticleName{The Modified Toda Hierarchy}

\Author{Wenjuan RUI~$^{\rm ab}$, Wenchuang GUAN~$^{\rm a}$, Yi YANG~$^{\rm c}$ and Jipeng CHENG~$^{\rm ab}$}
\AuthorNameForHeading{W.~Rui, W.~Guan, Y.~Yang and J.~Cheng}

\Address{$^{\rm a)}$~School of Mathematics, China University of Mining and Technology, Xuzhou, \\
\hphantom{$^{\rm a)}$}~Jiangsu 221116, P.R.~China}
\EmailD{\href{mailto:ruiwj@126.com}{ruiwj@126.com}, \href{mailto:wenchuangguan@163.com}{wenchuangguan@163.com}, \href{mailto:chengjp@cumt.edu.cn}{chengjp@cumt.edu.cn}, \newline
\hspace*{17.9mm}\href{mailto:chengjipeng1983@163.com}{chengjipeng1983@163.com}}

\Address{$^{\rm b)}$~Jiangsu Center for Applied Mathematics (CUMT), Xuzhou, Jiangsu 221116, P.R.~China}

\Address{$^{\rm c)}$~School of Mathematics, Sun Yat-sen University,
Guangzhou, Guangdong 510000, P.R.~China}
\EmailD{\href{mailto:yangy875@mail2.sysu.edu.cn}{yangy875@mail2.sysu.edu.cn}}

\ArticleDates{Received August 20, 2024, in final form December 05, 2024; Published online December 11, 2024}

\Abstract{In this paper, modified Toda (mToda) equation is generalized to form an integrable hierarchy in the framework of Sato theory, which is therefore called mToda hierarchy. Inspired by the fact that Toda hierarchy is 2-component generalization of usual KP hierarchy, mToda hierarchy is constructed from bilinear equations of 2-component first modified KP hierarchy, where we provide the corresponding equivalence with Lax formulations. Then it is demonstrated that there are Miura links between Toda and mToda hierarchies, which means the definition of mToda hierarchy here is reasonable. Finally, Darboux transformations of the Toda and mToda hierarchies are also constructed by using the aforementioned Miura links.}


\Classification{35Q51; 37K10; 37K40}

\section{Introduction}

The modified KP (mKP) hierarchy \cite{cheng2018jgp,Dickey1999lmp,Jimbo1983infin,Kac2018jjm,Kiso1990ptp,Kuper1985cmp,Kuper1995cmp} has obtained great success in mathematical physics and integrable systems, which is related to the famous KP hierarchy by Miura links~\mbox{\cite{Shaw1997, Yang2022}}. As one of the most important generalizations of KP hierarchy, Toda hierarchy~\mbox{\cite{Takasaki2018, Ueno1982}} also plays a key role in mathematical physics. For specific equations, Toda hierarchy contains the famous Toda equation, which has many important generalization \cite{daihh,hirota2004,Mikhailov1979,Mikhailov1981}. Among them, we are more interested in modified Toda (mToda) equation defined by
\begin{align*}
u(s)_y=u(s)(v(s)-v(s+1)),\qquad v(s)_x=v(s)(u(s)-u(s-1)),
\end{align*}
which is related with Toda equation by Miura transformation \cite{hirota2004}.
If further set $u(s)=\partial_x \varphi(s)$, $v(s)={\rm e}^{\varphi(s)-\varphi(s-1)}$, then mToda equation can be rewritten into exponential form
	\begin{align*}
		\pa_{x}\pa_{y}\varphi(s)+	\bigl({\rm e}^{\varphi(s+1)-\varphi(s)}-{\rm e}^{\varphi(s)-\varphi(s-1)}\bigr)\pa_{x}\varphi(s)=0.
	\end{align*}
Next, we expect to construct one integrable hierarchy containing this mToda equation, which will be called mToda hierarchy. Here the expected mToda hierarchy should be related with Toda hierarchy by Miura links just like the KP and mKP hierarchies. In fact, notice that Toda hierarchy is the 2-component generalization \cite{Jimbo1983infin,Liu2024,Ueno1982} of the usual KP hierarchy, while the KP hierarchy has Miura links \cite{Shaw1997, Yang2022} with the first mKP hierarchy. Thus, it is expected that mToda hierarchy is the 2-component first mKP hierarchy \cite{Jimbo1983infin, van2015}
\begin{align}
&{}{\rm Res}_ z \bigl(z^{s-s'-1}\tau_{0,s}\bigl(\mathbf{x}-\bigl[z^{-1}\bigr]_1\bigr)\tau_{1,s'}\bigl(\mathbf{x}'+\bigl[z^{-1}\bigr]_1\bigr){\rm e}^{\xi(\mathbf{x}^{(1)}-\mathbf{x}^{(1)\prime},z)} \nonumber\\
&\qquad{}+z^{s'-s-2} \tau_{0,s+1}\bigl(\mathbf{x}-\bigl[z^{-1}\bigr]_2\bigr)\tau_{1,s'-1}\bigl(\mathbf{x}'+\bigl[z^{-1}\bigr]_2\bigr){\rm e}^{\xi(\mathbf{x}^{(2)}-\mathbf{x}^{(2)\prime},z)}\bigr)
=\tau_{1,s}(\mathbf{x})\tau_{0,s'}(\mathbf{x}'),\!\!\label{2mkpbilineartau}
\end{align} 	
where ${\rm Res}_ z \sum_i a_iz^i=a_{-1}$, $\mathbf{x}=\bigl(\mathbf{x}^{(1)},\mathbf{x}^{(2)}\bigr)$, $\mathbf{x}^{(i)}=\bigl(x_1^{(i)},x_2^{(i)},\dots\bigr)$, $\bigl[z^{-1}\bigr]=\bigl(z^{-1},z^{-2}/2,\dots\bigr)$, \smash{$\mathbf{x}-\bigl[z^{-1}\bigr]_i=\bigl(\mathbf{x}^{(1)},\mathbf{x}^{(2)}\bigr)
\big|_{\mathbf{x}^{(i)}\mapsto\mathbf{x}^{(i)}-[z^{-1}]}$}, and \smash{$\xi\bigl(\mathbf{x}^{(i)},z\bigr)=\sum_{k\geq1}x_k^{(i)}z^k$}.
	
Here we will investigate the Lax structure of the 2-component first mKP hierarchy \eqref{2mkpbilineartau}. It is found that the corresponding Lax operators are expressed by pseudo-difference operators
\begin{align*}
&{}L_1=u_{-1}(s,\mathbf{x})\La+\sum_{i=0}^{\infty}u_i(s,\mathbf{x})\La^{-i},\qquad u_{-1}\neq 0,\\
&{}L_2=\bar{u}_{-1}(s,\mathbf{x})\La^{-1}+\sum_{i=0}^{\infty}\bar{u}_i(s,\mathbf{x})\La^{i},
\qquad \bar{u}_{-1}\neq 0, 
\end{align*}
satisfying the Lax equations
\begin{eqnarray}\label{2mkplax01}
\pa_{x_n^{(1)}}L_j=\bigl[\bigl(L_1^n\bigr)_{\Delta,\geq1},L_j\bigr], \qquad \pa_{x_n^{(2)}}L_j=\bigl[\bigl(L_2^n\bigr)_{\Delta^{*},\geq1},L_j\bigr],\qquad j=1,2,
\end{eqnarray}
where $\La$ is the shift operator defined by $\La(f(s))=f(s+1)$, $\Delta=\La-1$ and $\Delta^*=\La^{-1}-1$. One can refer to Section \ref{section2} for more details on the above symbols. Here we would like to comment that pseudo-difference operators are widely used in integrable systems \cite{Adler1999,Dickey1999lmp,Takasaki2018,zabrodin}. And note that we find that Lax equations \eqref{2mkplax01} contain the mToda equation, which means that 2-component first mKP hierarchy is just the desired mToda integrable hierarchy.

Notice that in Lax formulations, the initial term in $L_1$ for mToda is $u_{-1}(\mathbf{x})\La$ for some nonzero function $u_{-1}(\mathbf{x})$, while for Toda hierarchy, it is just $\Lambda$. Another differences are the flow generators. For mToda flow generators, they have the form $A_{\Delta,\geq 1}$ or $A_{\Delta^*,\geq 1}$, while in Toda case, they are $A_{\Lambda,\geq 0}$ or $A_{\Lambda,< 0}$. It is found here that mToda hierarchy is gauge equivalent to Toda hierarchy, which is called the Miura links, that is,
\begin{align*}
\xymatrix@R=0.1cm{
& {\boxed{\rm mToda}} \\
\text{anti-Miura} \colon \ {\boxed{\rm Toda}} \ar[ur]^{\mathcal{T}_1=q(s)^{-1}} \ar[dr]_{\mathcal{T}_2=\Delta^{-1}r(s)} \\
& {\boxed{\rm mToda}} }\qquad\xymatrix@R=0.1cm{
& {\boxed{\rm Toda}} \\
{\rm Miura}\colon \ {\boxed{\rm mToda}} \ar[ur]^{T_1=c_0^{-1}(s)} \ar[dr]_{T_2=c_0^{-1}(s+1)\Delta} \\
& {\boxed{\rm Toda} }}
\end{align*}
where $q(s)$ and $r(s)$ are Toda eigenfunction and adjoint eigenfunction, 	and $c_0(s)$ is the coefficient of $\La^0$ in $S_1$ defined by $L_1=S_1\Lambda S^{-1}_1$. The existence of Miura links confirms again that \eqref{2mkplax01} is the expected mToda hierarchy. Here we would like to point out that Miura links discussed here are just gauge transformations in integrable systems. And the concept of gauge transformations was firstly introduced in \cite{zakharov1,zakharov2}.

For mToda hierarchy \eqref{2mkplax01} discussed here, it has been used in various aspects in integrable system. Recently, mToda hierarchy has emerged in different types of constraints of Toda hierarchy, including constrained Toda hierarchy (C-Toda hierarchy) \cite{Krichever2022lmp} as well as the Toda lattice with the constraint of type B (B-Toda
hierarchy) \cite{guan2024, Krichever2023pd,Prokofev2023tmp}. Both C-Toda and B-Toda hierarchies are sub-hierarchies of mToda hierarchy. The mToda hierarchy can also be used to describe the spectral representation of Toda eigenfunction and Toda adjoint eigenfunction (see \cite{Liu2024}).

Besides the results mentioned before, we also show that the two mToda tau functions $\tau_{0,s}(\mathbf{x})$ and $\tau_{1,s}(\mathbf{x})$ are Toda tau functions, which can be linked by Toda eigenfunction or Toda adjoint eigenfunction. We also derive the mToda bilinear equation from the mToda Lax equations, and prove the existence of mToda tau functions. Based upon Miura links, the Toda and mToda Darboux transformations are also obtained by
\begin{align*}
{\rm Toda} \xrightarrow{\text{anti-Miura}} {\rm mToda} \xrightarrow{{\rm Miura}} {\rm Toda},\qquad {\rm mToda} \xrightarrow{{\rm Miura}} {\rm Toda} \xrightarrow{\text{anti-Miura}} {\rm mToda}.
\end{align*}
	
This paper is organized as follows. Firstly, in Section \ref{section2}, some notations and properties for formal difference operators are given, which will be used in discussion of mToda Lax equations. In Section \ref{section3}, starting from the bilinear equation, we derive the Lax formulations of the mToda hierarchy. We also prove that the bilinear equation can be obtained from the Lax equations. In~Section~\ref{section4}, two kinds of anti-Miura and Miura transformations between the Toda and mToda hierarchies are discussed. Then, in Section \ref{section5}, we discuss mToda tau functions and their relations with Toda tau functions. Finally, in Section \ref{section6}, the Darboux transformations of Toda and mToda hierarchies are constructed by means of Miura links.
	
\section{Formal difference operators}\label{section2}
In this section, let us review some important properties about formal difference operators, which will be used in the discussion of mToda Lax equations. Symbols involving $\Delta$ and $\Delta^*$ used here can be found in \cite{Liu2010}.
Let $\mathcal{A}$ denote the algebra of smooth complex functions in the indeterminate~$s$ and variables $\mathbf{x}$. The space $\mathcal{A}\bigl[\bigl[\La,\La^{-1}\bigr]\bigr]$ of formal difference operators consists of all expressions of the form
\begin{align*}
A=\sum_{m\in\mathbb{Z}}a_m(s)\La^m.
\end{align*}
Define the following symbols for above $A$:
\begin{align*}
A_{P}=\sum_{\text{$m$ satisfies $P$}}a_m(s)\La^m,\qquad A_{[k]}=a_k(s)\Lambda^k,\qquad A^*=\sum_{m\in \mathbb{Z}}\Lambda^{-m}a_m(s),
\end{align*}
where $P\in\{\geq k,\leq k,>k,<k\}$. Let $\mathcal{A}((\La))$ and $\mathcal{A}\bigl(\bigl(\La^{-1}\bigr)\bigr)$ denote the subspace of
$\mathcal{A}\bigl[\bigl[\La,\La^{-1}\bigr]\bigr]$, whose elements take
the following forms, respectively:
\begin{align*}
\sum_{m=m_0}^{\infty}a_m(s)\La^m\in\mathcal{A}((\La)), \qquad
\sum_{m=-\infty}^{n_0}a_m(s)\La^m\in\mathcal{A}\bigl(\bigl(\La^{-1}\bigr)\bigr),
\end{align*}
for some fixed integers $m_0$ and $n_0$.
Both $\mathcal{A}((\La))$ and $\mathcal{A}\bigl(\bigl(\La^{-1}\bigr)\bigr)$ are the associative
rings, where the multiplication is defined by
\begin{align*}
\bigl(f(s)\Lambda^i\bigr)\bigl(g(s)\Lambda^j\bigr)=f(s)g(s+i)\Lambda^{i+j}.
\end{align*}
Given $A\in\mathcal{A}\bigl(\bigl(\La^{\pm1}\bigr)\bigr)$ and $f\in\mathcal{A}$, we use $Af$ or $A\cdot f$ to denote the multiplication of $A$ with $f$, while $A(f)$ means the action of $A$ on $f$.

Introduce the difference operator $\Delta =\Lambda-1$ and its adjoint operator $\Delta^* =\Lambda^{-1}-1$, and define the following multiplication operations as follows \cite{Liu2010}, for $j\in \mathbb{Z}$:
\begin{align*}
&{}\Delta^j\cdot f(s)=\sum_{i=0}^\infty\binom{j}{i}(\Delta^i(f))(s+j-i)\Delta^{j-i},\\
&{}\Delta^{*j}\cdot f(s)=\sum_{i=0}^\infty\binom{j}{i}(\Delta^{*i}(f))(s+j-i)\Delta^{*j-i}.
\end{align*}
Then two associative rings $\mathcal{A}\bigl(\bigl(\Delta^{-1}\bigr)\bigr)$ and $\mathcal{A}\bigl(\bigl(\Delta^{*-1}\bigr)\bigr)$ are of the following forms:
\begin{align*}
\sum_{m=-\infty}^{k_0}b_m(s)\Delta^m\in\mathcal{A}\bigl(\bigl(\Delta^{-1}\bigr)\bigr),\qquad
\sum_{m=-\infty}^{l_0} b_m(s)\Delta^{*m}\in\mathcal{A}\bigl(\bigl(\Delta^*\bigr)^{-1}\bigr),
\end{align*}
for some fixed integers $k_0$ and $l_0$.
If define the following expansions
for $Q\in\big\{\Lambda,
\Lambda^{-1}\big\}$ and $R\in\{\Delta,\Delta^*\}$ by
\begin{align}
(Q-1)^{-1}=\sum_{j=1}^\infty Q^{-j},\qquad (R+1)^{-k}=\sum_{j=0}^\infty \binom{-k}{j}R^{-k-j},\qquad k>0,\label{delta-lambda}
\end{align}
then for arbitrary formal difference operator $A\in \mathcal{A}\bigl(\bigl(\Lambda^{-1}\bigr)\bigr)$ (resp.\ $\mathcal{A}((\Lambda))$), it can be transformed~into $A\in \mathcal{A}\bigl(\bigl(\Delta^{-1}\bigr)\bigr)$ \big(resp.\ \smash{$\mathcal{A}\bigl(\bigl(\Delta^*\bigr)^{-1}\bigr)$}\big) by \eqref{delta-lambda} and $\Lambda=\Delta+1$ \big(resp.\ $\Lambda^{-1}=\Delta^*+1$\big), and vice versa.
Then we have
\[
\mathcal{A}\bigl(\bigl(\Lambda^{-1}\bigr)\bigr)=\mathcal{A}\bigl(\bigl(\Delta^{-1}\bigr)\bigr),\qquad \mathcal{A}((\Lambda))=\mathcal{A}\bigl(\bigl(\Delta^{*}\bigr)^{-1}\bigr).
\]

Similar to $A_{P}$ with $P\in\{\geq k, \leq k,>k,<k\}$ and $A_{[k]}$, we can also define $A_{\Delta,P}$ (or $A_{\Delta^*,P}$) to be the part of $A$ satisfying property $P$ with respect to operator $\Delta$ (or $\Delta^*$), and $A_{\Delta,[k]}$ \big(or~$A_{\Delta^*,[k]}$\big) to be the part $\Delta^k$ \big(or~$\Delta^{*k}$\big). In what follows, we will denote $\iota_{\La^{\pm1}}A$ be the expansion of $A$ in~terms of operator $\Lambda^i$ in $\mathcal{A}\bigl(\bigl(\La ^{\pm1}\bigr)\bigr)$.
\begin{Lemma}\label{fdoprop1}
For any formal difference operator $A\in \mathcal{A}\bigl[\bigl[\La,\La^{-1}\bigr]\bigr]$,
\begin{alignat*}{3}
&A_{\ge0}=A_{\Delta,\ge0}, \qquad && A_{\le0}=A_{\Delta^*,\ge0}, &\\
&A_{\Delta,[0]}=A_{\ge0}(1), \qquad && A_{\Delta^*,[0]}=A_{\le0}(1).&
\end{alignat*}
\end{Lemma}
\begin{proof}
Firstly, by \eqref{delta-lambda}, we can find that $\Lambda^{-k}$, $k>0$ can not produce non-negative $\Delta$-powers, therefore
if assume $A=\sum_i a_i\Lambda^i$,
\begin{align*}
A_{\geq 0}=\sum_{i\geq 0}a_i\Lambda^i=\sum_{i\geq 0}a_i(\Delta+1)^i=A_{\Delta,\geq 0}.
\end{align*}
If further apply $A_{\geq 0}=A_{\Delta,\geq 0}$ to $1$, we can prove $A_{\Delta,[0]}=A_{\ge0}(1)$. Others can be similarly proved.
\end{proof}

\begin{Lemma}\label{proplambda2}
For any formal difference operator $A\in \mathcal{A}\bigl[\bigl[\La,\La^{-1}\bigr]\bigr]$, 	
\begin{align*}	
&{}\bigl(A\cdot\iota_{\La^{-1}}\Delta^{*-1}\bigr)_{\ge1}\cdot\Delta^*=A_{\ge1}-A_{\ge1}|_{\La=1}=A_{\Delta,\geq 1}, \\
&{}\bigl(A\cdot\iota_{\La^{-1}}\Delta^{*-1}\bigr)_{\le0}\cdot\Delta^*=A_{<0}+A_{\geq 0}|_{\La=1}=A_{\Delta,\leq 0}, \\
&{}\bigl(A\cdot\iota_{\La}\Delta^{*-1}\bigr)_{\geq 1}\cdot\Delta^*=A_{>0}+A_{\leq 0}|_{\La=1}=A_{\Delta^*,\leq 0}, \\
&{}\bigl(A\cdot\iota_{\La}\Delta^{*-1}\bigr)_{\leq 0}\cdot\Delta^*=A_{<0}-A_{<0}|_{\La=1}=A_{\Delta^*,\geq 1}.
\end{align*}
\end{Lemma}
\begin{proof}
Firstly, by Lemma \ref{fdoprop1},
\begin{align*}
\bigl(A\cdot\iota_{\La^{-1}}\Delta^{*-1}\bigr)_{\ge1}\cdot\Delta^*
&{}=\bigl(A\cdot\iota_{\La^{-1}}\Delta^{-1}\bigr)_{\ge0}\cdot\Delta=\bigl(A\cdot\iota_{\La^{-1}}\Delta^{-1}\bigr)_{\Delta,\ge0}\cdot\Delta\\
&{}=\bigl(A_{\Delta,\geq 1}\cdot\iota_{\La^{-1}}\Delta^{-1}\bigr)\cdot\Delta=A_{\Delta,\geq 1}.
\end{align*}
If subtract $A$ from the first relation, we can obtain the second one.
The third and fourth formulas can be similarly obtained.
\end{proof}	
	
\begin{Lemma}\label{fdoprop2}
For any formal difference operator $A\in \mathcal{A}\bigl[\bigl[\La,\La^{-1}\bigr]\bigr]$ and any $k\in\mathbb{Z}$,
\begin{align*}
&{}\Delta^k\cdot A\cdot\iota_{\Lambda^{\pm 1}}\Delta^{-k}=\Delta^{*k}\cdot \La^{k}(A)\cdot\iota_{\Lambda^{\pm 1}}\Delta^{*-k}, \\
&{}\bigl(\iota_{\Lambda}\Delta^{-k}\cdot A^{*}\cdot\Delta^k\bigr)_{\Delta^*,\ge k}=\iota_{\Lambda}\Delta^{-k}\cdot(A_{\Delta,\ge k})^{*}\cdot\Delta^k=\iota_{\Lambda^{-1}}\Delta^{-k}\cdot(A_{\Delta,\ge k})^{*}\cdot\Delta^k, \\
&{}\bigl(\iota_{\Lambda^{-1}}\Delta^{-k}\cdot A^{*}\cdot\Delta^k\bigr)_{\Delta,\ge k}=\iota_{\Lambda}\Delta^{-k}\cdot(A_{\Delta^{*},\ge k})^{*}\cdot\Delta^k=\iota_{\Lambda^{-1}}\Delta^{-k}\cdot(A_{\Delta^{*},\ge k})^{*}\cdot\Delta^k.
\end{align*}
\end{Lemma}
\begin{proof}
The first relation is derived by $\Delta^*=-\Lambda^{-1}\Delta$ and $\iota_{\Lambda^{\pm 1}}\Delta^{*-1}=-\Lambda\cdot\iota_{\Lambda^{\pm 1}}\Delta^{-1}$.
If assume $A=\sum_{m}a_m(s)\Delta^m$, then by $\iota_{\Lambda^{\pm1}}\Delta^{*-k}\cdot \Delta^{*k}=1$,
\begin{gather*}
\bigl(\iota_{\Lambda}\Delta^{-k}\cdot A^{*}\cdot\Delta^k\bigr)_{\Delta^*,\ge k}=\biggl(\bigl(\iota_{\Lambda}\Delta^{*-k}\La^{-k}\bigr)
\sum_{m}\Delta^{*m}a_m(s)\bigl(\La^k\Delta^{*k}\bigr)\biggr)_{\Delta^*,\ge k}\\
\hphantom{\bigl(\iota_{\Lambda}\Delta^{-k}\cdot A^{*}\cdot\Delta^k\bigr)_{\Delta^*,\ge k}}{}
=\sum_{m\ge k}\Delta^{*m-k}a_m(s-k)\Delta^{*k}, \\
\iota_{\Lambda^{\pm1}}\Delta^{-k}\cdot\bigl(A_{\Delta,\ge k}\bigr)^{*}\cdot\Delta^k=\iota_{\Lambda^{\pm1}}\bigl(\Delta^{*-k}\La^{-k}\bigr)\sum_{m\ge k}^{\infty}\Delta^{*m}a_m(s)\bigl(\La^k\Delta^{*k}\bigr)\\
\hphantom{\iota_{\Lambda^{\pm1}}\Delta^{-k}\cdot\bigl(A_{\Delta,\ge k}\bigr)^{*}\cdot\Delta^k}{}
=\sum_{m\ge k}^{\infty}\Delta^{*m-k}a_m(s-k)\Delta^{*k},
\end{gather*}
which imply the second relation. Similarly, we can prove the third formula.
\end{proof}	

\begin{Lemma}\label{fdoprop3}
For any formal difference operator $A\in \mathcal{A}\bigl[\bigl[\La,\La^{-1}\bigr]\bigr]$ and any function $f\in\mathcal{A}$,
\begin{gather*}
\bigl(f^{-1}\cdot A\cdot f\bigr)_{\Delta,\ge1} = f^{-1}\cdot A_{\ge0}\cdot f-f^{-1}\cdot A_{\ge0}(f) = f^{-1}\cdot A_{>0}\cdot f-f^{-1}\cdot A_{>0}(f), \\
\bigl(f^{-1}\cdot A\cdot f\bigr)_{\Delta^{*},\ge1} = f^{-1}\cdot A_{\le0}\cdot f-f^{-1}\cdot A_{\le0}(f) = f^{-1}\cdot A_{<0}\cdot f-f^{-1}\cdot A_{<0}(f), \\
\bigl(\iota_{\Lambda^{-1}}\Delta^{-1}\cdot f\cdot A\cdot f^{-1}\cdot\Delta\bigr)_{\Delta,\ge1} \\
\qquad{} = \iota_{\Lambda^{\pm1}}\Delta^{-1}\cdot f\cdot A_{\ge0}\cdot f^{-1}\cdot\Delta
-\iota_{\Lambda^{\pm1}}\Delta^{-1}\cdot(A_{\ge0})^{*}(f)\cdot f^{-1}\cdot\Delta \\
\qquad{} = \iota_{\Lambda^{\pm1}}\Delta^{-1}\cdot f\cdot A_{>0}\cdot f^{-1}\cdot\Delta-\iota_{\Lambda^{\pm1}}\Delta^{-1}\cdot(A_{>0})^{*}(f)\cdot f^{-1}\cdot\Delta, \\
\bigl(\iota_{\Lambda}\Delta^{*-1}\cdot f\cdot A\cdot f^{-1}\cdot\Delta^{*}\bigr)_{\Delta^{*},\ge1}\\
\qquad{} = \iota_{\Lambda^{\pm1}}\Delta^{*-1}\cdot f\cdot A_{\le0}\cdot f^{-1}\cdot\Delta^{*}
-\iota_{\Lambda^{\pm1}}\Delta^{*-1}\cdot(A_{\le0})^{*}(f)\cdot f^{-1}\cdot\Delta \\
\qquad{} = \iota_{\Lambda^{\pm1}}\Delta^{*-1}\cdot f\cdot A_{<0}\cdot f^{-1}\cdot\Delta^{*}-\iota_{\Lambda^{\pm1}}\Delta^{*-1}\cdot(A_{<0})^{*}(f)\cdot f^{-1}\cdot\Delta^{*}.
\end{gather*}
\end{Lemma}
\begin{proof}
By using Lemma \ref{fdoprop1}, we can get
\begin{align*}
\bigl(f^{-1}\cdot A\cdot f\bigr)_{\Delta,\ge1}
&{}=\bigl(f^{-1}\cdot A\cdot f\bigr)_{\Delta,\ge0}-\bigl(f^{-1}\cdot A\cdot f\bigr)_{\Delta,[0]}\\
&{}=\bigl(f^{-1}\cdot A\cdot f\bigr)_{\ge0}-\bigl(f^{-1}\cdot A\cdot f\bigr)_{\ge0}|_{\La=1} \\
&{}=f^{-1}\cdot A_{\ge0}\cdot f-f^{-1}\cdot A_{\ge0}(f) \\
&{}=f^{-1}\cdot A_{>0}\cdot f-f^{-1}\cdot A_{>0}(f).
\end{align*}
Similarly, we can get the second formula. As for the third formula, it can be derived by the third relation in Lemma \ref{fdoprop2} with $k=1$ and formula for $\bigl(f^{-1}\cdot A\cdot f\bigr)_{\Delta^{*},\ge1}$.
As for the fourth formula, it can be proved by similar to the third one.
\end{proof}		

\section{Lax formulations of mToda hierarchy}\label{section3}
In this section, we first derive the Lax formulations of the mToda hierarchy from the bilinear equation (\ref{2mkpbilineartau}). Then, we also provide the proof of bilinear equation obtained from the Lax formulations.

For 2-component first mKP hierarchy \eqref{2mkpbilineartau}, let us introduce wave functions $\Psi_j(s,\mathbf{x},z)$, ${j\!=\!1,2}$, and their adjoint $\Psi_j^*(s,\mathbf{x},z)$ in terms of the tau function as follows:
\begin{align}
&{}\Psi_1(s,\mathbf{x},z)=\frac{\tau_{0,s}\bigl(\mathbf{x}-\bigl[z^{-1}\bigr]_1\bigr)}{\tau_{1,s}(\mathbf{x})}z^{s}{\rm e}^{\xi(\mathbf{x}^{(1)},z)}=w^{(1)}(s,\mathbf{x},z)z^s {\rm e}^{\xi(\mathbf{x}^{(1)},z)}, \label{defwavefun1} \\ 	
&{}\Psi_1^*(s,\mathbf{x},z)=\frac{\tau_{1,s}\bigl(\mathbf{x}+\bigl[z^{-1}\bigr]_1\bigr)}{\tau_{0,s}(\mathbf{x})}z^{-s}{\rm e}^{-\xi(\mathbf{x}^{(1)},z)}=w^{(1)*}(s,\mathbf{x},z)z^{-s}{\rm e}^{-\xi(\mathbf{x}^{(1)},z)}, \label{defwavefun2} \\
&{}\Psi_2(s,\mathbf{x},z)=\frac{\tau_{0,s+1}(\mathbf{x}-[z]_2)}{\tau_{1,s}(\mathbf{x})}z^{s}{\rm e}^{\xi(\mathbf{x}^{(2)},z^{-1})}=w^{(2)}(s,\mathbf{x},z)z^{s}{\rm e}^{\xi(\mathbf{x}^{(2)},z^{-1})}, \label{addefwavefun1} \\	
&{}\Psi_2^*(s,\mathbf{x},z)=\frac{\tau_{1,s-1}(\mathbf{x}+[z]_2)}{\tau_{0,s}(\mathbf{x})}z^{-s+1}{\rm e}^{-\xi(\mathbf{x}^{(2)},z^{-1})}=w^{(2)*}(s,\mathbf{x},z)z^{-s}{\rm e}^{-\xi(\mathbf{x}^{(2)},z^{-1})},\label{addefwavefun2}
\end{align}
where $w^{(j)}(s,\mathbf{x},z)$ and $w^{(j)*}(s,\mathbf{x},z)$ are formal power series
with respect to $z^{-1}$ (for $j=1$) and~$z$ (for $j=2$)
\begin{alignat*}{3}
&w^{(1)}(s,\mathbf{x},z)=\sum_{i=0}^{\infty}c_i(s,\mathbf{x})z^{-i}, \qquad &&
w^{(2)}(s,\mathbf{x},z)=\sum_{i=0}^{\infty}\bar{c}_i(s,\mathbf{x})z^i, &\\
& w^{(1)*}(s,\mathbf{x},z)=\sum_{i=0}^{\infty}c'_i(s,\mathbf{x})z^{-i}, \qquad &&	w^{(2)*}(s,\mathbf{x},z)=\sum_{i=0}^{\infty}\bar{c}'_i(s,\mathbf{x})z^{i+1}, &
\end{alignat*}
with $c_0(s,\mathbf{x})\neq 0$ and $\bar{c}_0(s,\mathbf{x})\neq 0$.
Hence, the bilinear equation (\ref{2mkpbilineartau}) can be rewritten into
\begin{gather*}
\oint_{C_\infty}\frac{{\rm d}z}{2\pi {\rm i}z }\Psi_1(s,\mathbf{x},z)\Psi_1^*\bigl(s',\mathbf{x}',z\bigr)+\oint_{C_0}\frac{{\rm d}z}{2\pi {\rm i}z }\Psi_2(s,\mathbf{x},z)\Psi_2^*(s',\mathbf{x}',z)=1,
\end{gather*}
where $C_{\infty}$ denotes the circle around $z=\infty$, while $C_{0}$ denotes the circle around $z=0$. Both $C_{\infty}$ and $C_{0}$ are anticlockwise.
	
Next, introduce the wave operators $W_j$ and $\widetilde{W}_j$, $j=1,2$, as
\begin{alignat}{3} \label{defwmatrix}	
&W_1(s,\mathbf{x},\La)=S_1(s,\mathbf{x},\La){\rm e}^{\xi(\mathbf{x}^{(1)},\La)},\qquad && W_2(s,\mathbf{x},\La)=S_2(s,\mathbf{x},\La){\rm e}^{\xi(\mathbf{x}^{(2)},\La^{-1})}, & \notag\\
&\widetilde{W}_1(s,\mathbf{x},\La)=\widetilde{S}_1(s,\mathbf{x},\La){\rm e}^{-\xi(\mathbf{x}^{(1)},\La^{-1})}, \qquad &&
\widetilde{W}_2(s,\mathbf{x},\La)=\widetilde{S}_2(s,\mathbf{x},\La){\rm e}^{-\xi(\mathbf{x}^{(2)},\La)}, &
\end{alignat}
where $S_j$ and $\widetilde{S}_j$, $j=1,2$, are also called wave operators given by
\begin{alignat*}{3}
&S_1(s,\mathbf{x},\La)=\sum_{i\geq 0}c_i(s,\mathbf{x})\La^{-i},
\qquad && S_2(s,\mathbf{x},\La)=\sum_{i\geq 0}\bar{c}_i(s,\mathbf{x})\La^{i}, &\notag \\
&\widetilde{S}_1(s,\mathbf{x},\La)=\sum_{i\geq 0}c'_i(s,\mathbf{x})\La^{i},
\qquad && \widetilde{S}_2(s,\mathbf{x},\La)=\sum_{i\geq 0}\bar{c}'_i(s,\mathbf{x})\La^{-i-1}. &
\end{alignat*}
It can be clearly found from \eqref{defwavefun1}--\eqref{addefwavefun2} and \eqref{defwmatrix} that the wave functions and the adjoint wave functions are linked with wave operators in the following manner:
\begin{gather*}
\Psi_j(s,\mathbf{x},z)=W_j(s,\mathbf{x},\La)(z^s), \qquad
\Psi_j^*(s,\mathbf{x},z)=\widetilde{W}_j(s,\mathbf{x},\La)(z^{-s}),\qquad j=1,2.
\end{gather*}
To derive the Lax formulations of the mToda hierarchy, the following lemma is necessary.

\begin{Lemma}[{\cite{Adler1999}}]\label{proplambda1}
Let $A(s,\Lambda)=\sum_{i\in\mathbb{Z}}a_i(s)\Lambda^i$, $B(s,\Lambda)=\sum_{i\in\mathbb{Z}}b_i(s)\Lambda^i$are two pseudo-difference operators, then
\begin{equation*}
A(s,\Lambda)\cdot B^*(s,\Lambda)=\sum_{i\in \mathbb{Z}} {\rm Res}_{z}{z^{-1}}\bigl(A(s,\Lambda)\bigl(z^{\pm s}\bigr)\cdot B(s+i,\Lambda)\bigl(z^{\mp s\mp i}\bigr)\bigr)\Lambda^i.
\end{equation*}
\end{Lemma}
	
According to Lemma \ref{proplambda1}, one has
\begin{gather*}
W_1(s,\mathbf{x},\La)\widetilde{W}_1^{*}\bigl(s,\mathbf{x}',\La\bigr)+W_2(s,\mathbf{x},\La)\widetilde{W}_2^{*}\bigl(s,\mathbf{x}',\La\bigr)=\sum_{i\in\mathbb{Z}}\La^i,
\end{gather*}
namely,
\begin{gather}\label{2mkpwavematrix2}
S_1(s,\mathbf{x},\La){\rm e}^{\xi(\mathbf{x}^{(1)}-\mathbf{x}'^{(1)},\La)}\widetilde{S}_1^{*}\bigl(s,\mathbf{x}',\La\bigr)+S_2(s,\mathbf{x},\La){\rm e}^{\xi(\mathbf{x}^{(2)}-\mathbf{x}'^{(2)},\La^{-1})}\widetilde{S}_2^{*}\bigl(s,\mathbf{x}',\La\bigr)=\sum_{i\in\mathbb{Z}}\La^i.\!\!
\end{gather}
	
\begin{Theorem}\label{mTodasato}
The wave operators have the relations
\begin{alignat}{3}
&{\widetilde{S}}_1=-\iota_{\La}\Delta^{-1}\bigl(S^*_1\bigr)^{-1},\qquad && {\widetilde{S}}_2=\iota_{\La^{-1}}\Delta^{-1}\bigl(S^*_2\bigr)^{-1}, &\label{S1S2}\\
&{\widetilde{W}}_1=-\iota_{\La^{-1}}\Delta^{-1}\bigl(W^*_1\bigr)^{-1},\qquad && {\widetilde{W}}_2=\iota_{\La}\Delta^{-1}\bigl(W^*_2\bigr)^{-1}, &\label{W1W2}
\end{alignat}
where $\iota_{\La}\Delta^{-1}=-\sum_{i\geq 0}\Lambda^i$ and $\iota_{\La^{-1}}\Delta^{-1}=\sum_{i\geq 1}\Lambda^{-i}$.
Moreover, they satisfy the following evolution equations:
\begin{alignat}{3}
&\pa_{x_n^{(1)}}{S}_1=-\bigl( S_1\La^nS_1^{-1}\bigr)_{\Delta,\leq0}{S}_1,\qquad && \pa_{x_n^{(1)}}{S}_2=\bigl( S_1\La^nS_1^{-1}\bigr)_{\Delta,\geq1}{S}_2,&\label{mTodaSx1} \\
&\pa_{x_n^{(2)}}{S}_1=\bigl( S_2\La^{-n}S_2^{-1}\bigr)_{\Delta^{*},\geq1}{S}_1,\qquad&&
\pa_{x_n^{(2)}}{S}_2=-\bigl( S_2\La^{-n}S_2^{-1}\bigr)_{\Delta^{*},\leq0}{S}_2,& \nonumber 
\\
&\pa_{x_n^{(1)}}{W}_j=\bigl( S_1\La^nS_1^{-1}\bigr)_{\Delta,\geq1}{W}_j,\qquad&&
\pa_{x_n^{(2)}}{W}_j=\bigl( S_2\La^{-n}S_2^{-1}\bigr)_{\Delta^{*},\geq1}{W}_j.& \nonumber 
\end{alignat}
\end{Theorem}	
\begin{proof}
Firstly, by substituting $\mathbf{x}'=\mathbf{x}$ into expression \eqref{2mkpwavematrix2} and comparing the non-positive and positive powers of $\La$ on both sides, we can obtain the following relations:
\begin{eqnarray}\label{relawaves}
S_1(s,\mathbf{x},\La)\widetilde{S}_1^{*}(s,\mathbf{x},\La)=-\iota_{\La^{-1}}\Delta^{*-1}, \qquad S_2(s,\mathbf{x},\La)\widetilde{S}_2^{*}(s,\mathbf{x},\La)=\iota_{\La}\Delta^{*-1},
\end{eqnarray}
which implies \eqref{S1S2} and \eqref{W1W2}. Next, by differentiating (\ref{2mkpwavematrix2}) with respect to $x^{(1)}_n$ and letting $\mathbf{x}'=\mathbf{x}$, we have
\begin{eqnarray}\label{2mkpsatopf1}
\pa_{x^{(1)}_n}S_1\cdot\widetilde{S}_1^{*}+S_1\La^n\widetilde{S}_1^{*}+\pa_{x^{(1)}_n}S_2\cdot\widetilde{S}_2^{*}=0.
\end{eqnarray}
Hence, taking the positive power and non-positive power of $\La$ in (\ref{2mkpsatopf1}) and making use of (\ref{relawaves}), we can get, respectively,
\begin{align*}
\pa_{x^{(1)}_n}S_2=\bigl(S_1\La^nS_1^{-1}\iota_{\La^{-1}}\Delta^{*-1}\bigr)_{\geq 1}\Delta^* S_2,\qquad \pa_{x^{(1)}_n}S_1 =-\bigl(S_1\La^nS_1^{-1}\iota_{\La^{-1}}
\Delta^{*-1}\bigr)_{\leq 0}\Delta^*S_1.
\end{align*}
Then by Lemma \ref{proplambda2}, we have
\begin{align*}
&\pa_{x^{(1)}_n}S_2=\bigl(\bigl(S_1\La^nS_1^{-1}\bigr)_{\geq 1}-\bigl(S_1\La^nS_1^{-1}\bigr)_{\geq 1}|_{\La=1} \bigr)S_2=\bigl( S_1\La^nS_1^{-1}\bigr)_{\Delta,\geq1}S_2,\\
&\pa_{x^{(1)}_n}S_1=-\bigl(\bigl(S_1\La^nS_1^{-1}\bigr)_{<0}+\bigl(S_1\La^nS_1^{-1}\bigr)_{\geq 0}|_{\La=1} \bigr)S_1=-\bigl( S_1\La^nS_1^{-1}\bigr)_{\Delta,\leq0}S_1.
\end{align*}
Similarly, we can obtain $\pa_{x^{(2)}_n}S_i$.
\end{proof}

Further introduce the Lax operators $L_1$ and $L_2$ of mToda hierarchy as
\begin{gather*}
L_1=W_1\La W_1^{-1}=\sum_{i=-1}^{\infty}u_i(s,\mathbf{x})\La^{-i}\in\mathcal{A}\bigl(\bigl(\La^{-1}\bigr)\bigr), \\
L_2=W_2\La^{-1}W_2^{-1}=\sum_{i=-1}^{\infty}\bar{u}_i(s,\mathbf{x})\La^{i}\in\mathcal{A}((\La)),
\end{gather*}
then Lax equations follow from Theorem \ref{mTodasato}
\begin{eqnarray}\label{2mkplax}
\pa_{x_n^{(1)}}L_j=\bigl[B_n^{(1)},L_j\bigr], \qquad \pa_{x_n^{(2)}}L_j=\bigl[B_n^{(2)},L_j\bigr],\qquad j=1,2,
\end{eqnarray}
where \smash{$B_n^{(1)}=\bigl( S_1\La^nS_1^{-1}\bigr)_{\Delta,\geq1}$},
		\smash{$B_n^{(2)}=\bigl( S_2\La^{-n}S_2^{-1}\bigr)_{\Delta^{*},\geq1}$}.
\begin{Corollary}
Wave functions $\Psi_j$ and the adjoint wave functions $\Psi^*_j$, $j=1,2$, satisfy the auxiliary linear equations
	\begin{alignat*}{3}
		&{}L_1(\Psi_1)=z\Psi_1,\qquad&& L_2(\Psi_2)=z^{-1}\Psi_2, &\\
 &{}\pa_{x_n^{(1)}}\Psi_j=B_n^{(1)}(\Psi_j),\qquad&& \pa_{x_n^{(2)}}\Psi_j=B_n^{(2)}(\Psi_j),&\\
		&{}\bigl(\iota_{\La^{-1}}\Delta^{-1}L_1^*\Delta\bigr)(\Psi^*_1)=z\Psi^*_1,\qquad&& \bigl(\iota_{\La}\Delta^{-1}L_2^*\Delta\bigr)(\Psi^*_2)=z^{-1}\Psi^*_2, & \\
		&{}\pa_{x_n^{(1)}}\Psi^*_j=-\bigl(\iota_{\La^{-1}}\Delta^{-1}B_n^{(1)*}\Delta\bigr)(\Psi^*_j), \qquad && \pa_{x_n^{(2)}}\Psi^*_j=-\bigl(\iota_{\La}\Delta^{-1}B_n^{(2)*}\Delta\bigr)(\Psi^*_j).&
	\end{alignat*}
\end{Corollary}	
	
The following proposition states that mToda is consistent, that is, \smash{$\bigl[\pa_{x_n^{(i)}},\pa_{x_m^{(j)}}\bigr]=0$}.
\begin{Proposition}
$B_m^{(j)}$ satisfies the following relation:
\begin{eqnarray}\label{2mkpzs1}			
\pa_{x_n^{(i)}}B_m^{(j)}-\pa_{x_m^{(j)}}B_n^{(i)}+\bigl[B_m^{(j)},B_n^{(i)}\bigr]=0,\qquad i,j=1,2.
\end{eqnarray}
\end{Proposition}
\begin{proof}
Here we only prove the case of $i=1$, $j=2$, other cases are similar.
Firstly, denote
\[
A_{i,j}=\pa_{x_n^{(i)}}B_m^{(j)}-\pa_{x_m^{(j)}}B_n^{(i)}+\bigl[B_m^{(j)},B_n^{(i)}\bigr],
\]
then by (\ref{2mkplax}), $B^{(1)}_n=(L_1^n)_{\geq 1}-(L_1^n)_{\geq 1 }|_{\La=1}$ and $B^{(2)}_n=(L_2^m)_{<0}-(L_2^m)_{<0}|_{\La=1}$, we have
\begin{align*}
(A_{1,2})_{>0}=-\pa_{x_m^{(2)}}(L_1^n)_{>0}+\bigl[B_m^{(2)},B_n^{(1)}\bigr]_{>0}
=\bigl(-\pa_{x_m^{(2)}}(L_1^n)+\bigl[B_m^{(2)},L_1^n\bigr]\bigr)_{>0}=0.
\end{align*}
Similarly, we can prove $(A_{1,2})_{<0}$. Therefore, by $A_{1,2}(1)=0$,
\begin{align*}
(A_{1,2})_{[0]}=-((A_{1,2})_{>0}+(A_{1,2})_{<0})(1)=0,
\end{align*}
which means $A_{1,2}=0$.
\end{proof}

\begin{Example}
Let us give some explicit examples of nonlinear differential-difference equations of mToda hierarchy. Taking $m=n=1$, we have
\begin{align*}
B_1^{(1)}=u_{-1}(s,\mathbf{x})\La-u_{-1}(s,\mathbf{x}), \qquad B_1^{(2)}=\bar{u}_{-1}(s,\mathbf{x})\La^{-1}-\bar{u}_{-1}(s,\mathbf{x}),	
\end{align*}
then it follows from (\ref{2mkpzs1}) that
\begin{align}
&\pa_{x_1^{(1)}}\bar{u}_{-1}(s,\mathbf{x})+\bar{u}_{-1}(s,\mathbf{x})(u_{-1}(s,\mathbf{x})-u_{-1}(s-1,\mathbf{x}))=0, \label{ueq1} \\
&-\pa_{x_1^{(2)}}u_{-1}(s,\mathbf{x})+u_{-1}(s,\mathbf{x})(\bar{u}_{-1}(s+1,\mathbf{x})-\bar{u}_{-1}(s,\mathbf{x}))=0. \label{ueq2} 	
\end{align}
Equations (\ref{ueq1}) and (\ref{ueq2}) are just the mToda equation \cite{daihh, hirota2004} mentioned in the introduction.
\end{Example}

Conversely, if we start from the evolution equation of wave operators
\begin{alignat*}{3}
&\pa_{x_n^{(1)}}{S}_1=-\bigl( S_1\La^nS_1^{-1}\bigr)_{\Delta,\leq0}{S}_1,\qquad&&
\pa_{x_n^{(1)}}{S}_2=\bigl( S_1\La^nS_1^{-1}\bigr)_{\Delta,\geq1}{S}_2, &\\
&\pa_{x_n^{(2)}}{S}_1=\bigl( S_2\La^{-n}S_2^{-1}\bigr)_{\Delta^{*},\geq1}{S}_1,\qquad&&
\pa_{x_n^{(2)}}{S}_2=-\bigl( S_2\La^{-n}S_2^{-1}\bigr)_{\Delta^{*},\leq0}{S}_2,&\\
&\pa_{x_n^{(1)}}{W}_j=\bigl( S_1\La^nS_1^{-1}\bigr)_{\Delta,\geq1}{W}_j,\qquad&&
\pa_{x_n^{(2)}}{W}_j=\bigl( S_2\La^{-n}S_2^{-1}\bigr)_{\Delta^{*},\geq1}{W}_j,&
\end{alignat*}
then we have the following theorem. 	
\begin{Theorem}
Given wave operators satisfying above relations and set
\begin{alignat}{3}
&\Psi_1(s,\mathbf{x},z)=W_1(z^s),\qquad&&			
\Psi_1^*(s,\mathbf{x},z)=-\iota_{\La^{-1}}\Delta^{-1}(W^*_1)^{-1}(z^{-s}),&\nonumber\\
&\Psi_2(s,\mathbf{x},z)=W_2(z^s),\qquad &&
\Psi_2^*(s,\mathbf{x},z)=\iota_{\La}\Delta^{-1}(W^*_2)^{-1}(z^{-s}),&\label{mTodawave}
\end{alignat}
then $\Psi_j(s,\mathbf{x},z)$ and $\Psi_j^*(s,\mathbf{x},z)$, $j=1,2$, satisfy the bilinear identity
\begin{eqnarray}\label{Todabili}
\oint_{C_\infty}\frac{{\rm d}z}{2\pi {\rm i}z }\Psi_1(s,\mathbf{x},z)\Psi_1^*\bigl(s',\mathbf{x}',z\bigr)+\oint_{C_0}\frac{{\rm d}z}{2\pi {\rm i}z }\Psi_2(s,\mathbf{x},z)\Psi_1^*\bigl(s',\mathbf{x}',z\bigr)=1.
\end{eqnarray}
\end{Theorem}	
\begin{proof}
Firstly, it is obvious that
\begin{align*}
-W_1(s,\mathbf{x},\La)W^{-1}_1(s,\mathbf{x},\La)\iota_{\La^{-1}}\Delta^{*-1}+W_2(s,\mathbf{x},\La)W^{-1}_2(s,\mathbf{x},\La)\iota_{\La}\Delta^{*-1}=\sum_{i\in\mathbb{Z}}\La^i.
\end{align*}
It follows from evolution equations of $W_j$ that
\begin{gather*}
\pa_{x^{(1)}_n} W_1\cdot W_1^{-1}=\pa_{x^{(1)}_n}W_2\cdot W_2^{-1}, \qquad
\pa_{x^{(2)}_n}W_1\cdot W_1^{-1}=\pa_{x^{(2)}_n}W_2\cdot W_2^{-1}.
\end{gather*}
Therefore, by induction on the order of derivatives with $\mathbf{x}^{(1)}$ and $\mathbf{x}^{(2)}$, we can get for $\al=(\al_1,\al_2,\dots)\geq0$, $\beta=(\beta_1,\beta_2,\dots)\geq0$,
\begin{gather*}
-\pa_{\mathbf{x}^{(1)}}^\al \pa_{\mathbf{x}^{(2)}}^\beta W_1\cdot W_1^{-1}\iota_{\La^{-1}}\Delta^{*-1}+\pa_{\mathbf{x}^{(1)}}^\al \pa_{\mathbf{x}^{(2)}}^\beta W_2\cdot W_2^{-1}\iota_{\La}\Delta^{*-1}=
\begin{cases}
\sum\limits_{j\in \mathbb{Z}}\Lambda^j, & (\alpha,\beta)=(\mathbf{0},\mathbf{0}),\\
0, & (\alpha,\beta)\neq(\mathbf{0},\mathbf{0}),
\end{cases}
\end{gather*}
where
\[
\pa_{\mathbf{x}^{(1)}}^\al=\pa_{x^{(1)}_1}^{\al_1}\pa_{x^{(1)}_2}^{\al_2}\cdots ,	\qquad \text{and}\qquad \pa_{\mathbf{x}^{(2)}}^\beta=\pa_{x^{(2)}_1}^{\beta_1}\pa_{x^{(2)}_2}^{\beta_2}\cdots.
\]
By making use of Taylor expansion, finally it can be found that
\begin{align*}
-W_1(s,\mathbf{x},\La)W_1(s,\mathbf{x}',\La)^{-1}\iota_{\La^{-1}}\Delta^{*-1}+W_2(s,\mathbf{x},\La)W_2(s,\mathbf{x}',\La)^{-1}\iota_{\La}\Delta^{*-1}=\sum_{i\in\mathbb{Z}}\La^i,
\end{align*}
which implies (\ref{Todabili}) by applying Lemma \ref{proplambda1}.
\end{proof}
	
\section[Anti-Miura and Miura transformations between Toda and mToda hierarchies]{Anti-Miura and Miura transformations between Toda \\ and mToda hierarchies}\label{section4}

In this section, we first briefly introduce some basic facts about the Toda hierarchy and the automorphism of Toda and mToda hierarchies. Based on these, we will discuss the anti-Miura and Miura transformations between the Toda and mToda hierarchies. In this paper, we adopt the following convention, the Miura transformation refer to the transformation from mToda to Toda, whereas the anti-Miura transformation is the one from Toda to mToda.

\subsection{Basic facts about Toda hierarchy and its self-transformation}
Recall that Toda hierarchy \cite{Takasaki2018, Ueno1982} is defined by the following Lax equations:
\begin{align}
\pa_{x^{(1)}_n}\mathcal{L}_j=\bigl[(\mathcal{L}_1^n)_{\geq0},\mathcal{L}_j\bigr],
\qquad \pa_{x^{(2)}_n}{\mathcal{L}_j}=\bigl[({\mathcal{L}^n_2})_{<0},{\mathcal{L}_j}\bigr],\label{todalax}
\end{align}
where Toda Lax operators are given by
\[
\mathcal{L}_1(s,\mathbf{x},\La)=\La+\sum_{i=0}^{\infty}v_i(s,\mathbf{x})\La^{-i},\qquad \mathcal{L}_2(s,\mathbf{x},\La)=\bar{v}_{-1}(s,\mathbf{x})\La^{-1}+\sum_{i=0}^{\infty}\bar{v}_i(s,\mathbf{x})\La^{i}.
\]

Let $\mathcal{S}_1$ and $\mathcal{S}_1$ be the Toda wave operators of the form
\[
\mathcal{S}_1(s,\mathbf{x},\La)=1+\sum_{i=1}^{\infty}w_i(s,\mathbf{x})\Lambda^{-i},\qquad{\mathcal{S}_2}(s,\mathbf{x},\La)=\sum_{i=0}^{\infty}\bar{w}_i(s,\mathbf{x})\Lambda^{i},
\]
by which the Lax operators can be expressed as
\begin{align*}
\mathcal{L}_1=\mathcal{S}_1\La \mathcal{S}_1^{-1},\qquad\mathcal{L}_2=\mathcal{S}_2\La^{-1}\mathcal{S}_2^{-1}.
\end{align*}
Then the Lax equations (\ref{todalax}) are equivalent to the following Sato equations:
\begin{alignat}{3}
&\pa_{x^{(1)}_n}\mathcal{S}_1=-(\mathcal{L}^n_1)_{<0}\mathcal{S}_1,\qquad&&
\pa_{x^{(2)}_n}\mathcal{S}_1=(\mathcal{L}^n_2)_{<0}\mathcal{S}_1,& \nonumber\\
&\pa_{x^{(1)}_n}\mathcal{S}_2=(\mathcal{L}^n_1)_{\geq0}\mathcal{S}_2, \qquad&&
		\pa_{x^{(2)}_n}\mathcal{S}_2=-(\mathcal{L}^n_2)_{\geq0}\mathcal{S}_2.&\label{todasatos2}
\end{alignat}

The wave functions $\Phi_j$ and adjoint wave functions $\Phi^*_j$, $j=1,2$, of the Toda hierarchy are	defined as
\begin{alignat}{3}
&\Phi_1(s,\mathbf{x},z)=\mathcal{W}_1(s,\mathbf{x},\Lambda)(z^s),\qquad&&
\Phi_2(s,\mathbf{x},z)=\mathcal{W}_2(s,\mathbf{x},\Lambda)(z^{s}),&\notag \\
&\Phi^*_1(s,\mathbf{x},z)=\bigl(\mathcal{W}_1^{-1}(s,\mathbf{x},\Lambda)\bigr)^*(z^{-s}),\qquad&&
\Phi^*_2(s,\mathbf{x},z)=\bigl(\mathcal{W}_2^{-1}(s,\mathbf{x},\Lambda)\bigr)^*(z^{-s}),&\label{todawave}
\end{alignat}
where $\mathcal{W}_1(s,\mathbf{x},\Lambda)=
\mathcal{S}_1(s,\mathbf{x},\Lambda){\rm e}^{\xi(\mathbf{x}^{(1)},\Lambda)}$ and
$\mathcal{W}_2(s,\mathbf{x},\Lambda)=
\mathcal{S}_2(s,\mathbf{x},\Lambda){\rm e}^{\xi(\mathbf{x}^{(2)},\Lambda^{-1})}$.
It can be verified that wave functions $\Phi_j$ and adjoint wave functions $\Phi^*_j$, $j=1,2$, satisfy the following auxiliary linear equations:
\begin{alignat*}{5}
&\mathcal{L}_1(\Phi_1)=z\Phi_1,\quad&&
\mathcal{L}_2(\Phi_2)=z^{-1}\Phi_2,\quad&&
\pa_{x^{(1)}_n}\Phi_i=(\mathcal{L}^n_1)_{\geq0}(\Phi_i),\quad&&
\pa_{x^{(2)}_n}\Phi_i=(\mathcal{L}^n_2)_{<0}(\Phi_i),&\\
&\mathcal{L}^*_1(\Phi^*_1)=z\Phi^*_1,\quad \ &&
\mathcal{L}^*_2(\Phi^*_2)=z^{-1}\Phi^*_2,\quad \ &&
\pa_{x^{(1)}_n}\Phi^*_i=-(\mathcal{L}^n_1)_{\geq0}^*(\Phi^*_i),\quad \ &&
\pa_{x^{(2)}_n}\Phi^*_i=-(\mathcal{L}^n_2)_{<0}^*(\Phi^*_i),&
\end{alignat*}
and the bilinear equation
\begin{eqnarray}\label{todabiliwave}
\oint_{C_\infty}\frac{{\rm d}z}{2\pi {\rm i}z }\Phi_1(s,\mathbf{x},z)\Phi_1^*\bigl(s',\mathbf{x}',z\bigr)=\oint_{C_0}\frac{{\rm d}z}{2\pi {\rm i}z }\Phi_2(s,\mathbf{x},z)\Phi_2^*\bigl(s',\mathbf{x}',z\bigr).
\end{eqnarray}
The Toda wave functions and adjoint wave functions can be generalized to the eigenfunction~$q(s)$ and adjoint eigenfunction $r(s)$ of the Toda hierarchy defined by
\begin{alignat}{3}
&\pa_{x^{(1)}_n}q(s)=(\mathcal{L}_1^n)_{\geq0}(q(s)),\qquad&&
\pa_{x^{(2)}_n}q(s)=(\mathcal{L}_2^n)_{<0}(q(s)),&
\label{todaeigen} \\
&\pa_{x^{(1)}_n}r(s)=-((\mathcal{L}_1^n)_{\geq0})^*(r(s)),\qquad&&
\pa_{x^{(2)}_n}r(s)=-((\mathcal{L}_2^n)_{<0})^*(r(s)).&\nonumber 
\end{alignat}

There exists one tau function \cite{Ueno1982} \smash{$\tau^{{\rm Toda}}_s(\mathbf{x})$} such that wave functions $\Phi_j$ and adjoint wave functions $\Phi^*_j$, $j=1,2$, can be expressed in terms of tau functions $\tau^{{\rm Toda}}_s(\mathbf{x})$ as
\begin{align}
&\Phi_1(s,\mathbf{x},z)=\frac{\tau^{{\rm Toda}}_s
\bigl(\mathbf{x}-\bigl[z^{-1}\bigr]_1\bigr)}{\tau^{{\rm Toda}}_s(\mathbf{x})}{\rm e}^{\xi(\mathbf{x}^{(1)},z)}z^s,\nonumber\\
&{\Phi_2}(s,\mathbf{x},z)=\frac{\tau^{{\rm Toda}}_{s+1}(\mathbf{x}-[z]_2)}{\tau^{{\rm Toda}}_s(\mathbf{x})}
{\rm e}^{\xi(\mathbf{x}^{(2)},z^{-1})}z^{s},\label{TLtau}
\\
&{\Phi_1}^*(s,\mathbf{x},z)=\frac{\tau^{{\rm Toda}}_{s+1}\bigl(\mathbf{x}+\bigl[z^{-1}\bigr]_1\bigr)}{\tau^{{\rm Toda}}_{s+1}(\mathbf{x})}{\rm e}^{-\xi(\mathbf{x}^{(1)},z)}z^{-s},\nonumber\\
&{\Phi_2}^*(s,\mathbf{x},z)=\frac{\tau^{{\rm Toda}}_{s}(\mathbf{x}+[z]_2)}{\tau^{{\rm Toda}}_{s+1}(\mathbf{x})}{\rm e}^{-\xi(\mathbf{x}^{(2)},z^{-1})}z^{-s}.\label{todaawavetau}
\end{align}
Then the bilinear equation (\ref{todabiliwave}) can be expressed by Toda tau function
\begin{align*}
&\oint_{C_\infty}\frac{{\rm d}z}{2\pi {\rm i} }\tau^{{\rm Toda}}_s\bigl(\mathbf{x}-\bigl[z^{-1}\bigr]_1\bigr)\tau^{{\rm Toda}}_{s'}\bigl(\mathbf{x}'+\bigl[z^{-1}\bigr]_1\bigr)z^{s-s'}{\rm e}^{\xi(\mathbf{x}^{(1)}-\mathbf{x}^{(1)\prime},z)} \notag \\
&\qquad{}=\oint_{C_0}\frac{{\rm d}z}{2\pi {\rm i} }\tau^{{\rm Toda}}_{s+1}\bigl(\mathbf{x}-[z]_2\bigr)\tau^{{\rm Toda}}_{s'-1}\bigl(\mathbf{x}'+[z]_2\bigr)z^{s-s'}{\rm e}^{\xi(\mathbf{x}^{(2)}-\mathbf{x}^{(2)\prime},z^{-1})}.
\end{align*}

\begin{Proposition}\label{automorToda}
Assume $\mathcal{L}_{j}$ and $\mathcal{S}_{j}$, $j=1,2$, are Toda Lax operators and wave operators, respectively, and consider the following maps:
\smash{$\pi_0\colon \bigl(\mathcal{L}_i,\mathcal{S}_i,x^{(1)}_n,x^{(2)}_n\bigr)\mapsto
\bigl(\tilde{\mathcal{L}}_i,\tilde{\mathcal{S}}_i,\tilde{x}^{(1)}_n,\tilde{x}^{(2)}_n\bigr)$}
with
\begin{align*}
\tilde{\mathcal{L}}_i=\bar{w}_0\mathcal{L}^*_{3-i}\bar{w}^{-1}_0,\qquad \tilde{\mathcal{S}}_i=\bar{w}_0\bigl(\mathcal{S}^{-1}_{3-i}\bigr)^*,\qquad\tilde{x}^{(1)}_n=-x^{(2)}_n,\qquad\tilde{x}^{(2)}_n=-x^{(1)}_n,
\end{align*}
where $\bar{w}_0$ is the coefficient of $\La^0$ in wave function $\mathcal{S}_2$. Then $\pi_0$ is a self-transformation of the Toda hierarchy, that is,
\begin{alignat*}{3}
&\pa_{\tilde{x}_n^{(1)}}\tilde{\mathcal{L}}_j
=\bigl[\tilde{\mathcal{B}}^{(1)}_n,\tilde{\mathcal{L}}_j\bigr], \qquad&& \pa_{\tilde{x}_n^{(2)}}\tilde{\mathcal{L}}_j
=\bigl[\tilde{\mathcal{B}}^{(2)}_n,\tilde{\mathcal{L}}_j\bigr],&\\
&\pa_{\tilde{x}_n^{(1)}}{\tilde{\mathcal{S}}}_1=\tilde{\mathcal{B}}^{(1)}_n{\tilde{\mathcal{S}}}_1-{\tilde{\mathcal{S}}}_1\La^n,\qquad && \pa_{\tilde{x}_n^{(2)}}{\tilde{\mathcal{S}}}_1=\tilde{\mathcal{B}}^{(2)}_n{\tilde{\mathcal{S}}}_1, &\\		&\pa_{\tilde{x}_n^{(1)}}{\tilde{\mathcal{S}}}_2=\tilde{\mathcal{B}}^{(1)}_n{\tilde{\mathcal{S}}}_2,\qquad&&
\pa_{\tilde{x}_n^{(2)}}{\tilde{\mathcal{S}}}_2=\tilde{\mathcal{B}}^{(2)}_n{\tilde{\mathcal{S}}}_2-{\tilde{\mathcal{S}}}_2\La^{-n},&&
\end{alignat*}
where \smash{$\tilde{\mathcal{B}}^{(1)}_n=\bigl(\tilde{{\mathcal{L}}}^n_1\bigr)_{\geq0}$} and \smash{$\tilde{\mathcal{B}}^{(2)}_n=\bigl({\tilde{{\mathcal{L}}}^n_2}\bigr)_{<0}$}.
\end{Proposition}
\begin{proof}
Comparing coefficients of $\La^0$ on both sides	of \eqref{todasatos2}, we can get
\begin{align*}
\pa_{x^{(1)}_n}\bar{w}_0=(\mathcal{L}^n_1)_{[0]}\cdot\bar{w}_0, \qquad
\pa_{x^{(2)}_n}\bar{w}_0=-(\mathcal{L}^n_2)_{[0]}\cdot\bar{w}_0.
\end{align*}
Then by using $\bigl(A_{[0]}\bigr)^*=A_{[0]}$ and $(A_{<0})^*=(A^*)_{>0}$, we have the following relations:
\begin{gather*}
\tilde{\mathcal{B}}^{(1)}_n=\bigl(\bar{w}_0\mathcal{L}^{n*}_{2}\bar{w}^{-1}_0\bigr)_{\geq0}
=\bar{w}_0(\mathcal{L}^{n*}_{2})_{[0]}\bar{w}^{-1}_0
+\bigl(\bar{w}_0\mathcal{L}^{n*}_{2}\bar{w}^{-1}_0\bigr)_{>0}
=(\mathcal{L}^{n}_{2})_{[0]}
+\bar{w}_0(\mathcal{L}^{n*}_{2})_{>0}\bar{w}^{-1}_0\\
\hphantom{\tilde{\mathcal{B}}^{(1)}_n}{}
=-\pa_{x^{(1)}_n}\bar{w}_0\cdot \bar{w}_0^{-1} +\bar{w}_0\mathcal{B}^{(2)*}_n\bar{w}^{-1}_0,\\			\tilde{\mathcal{B}}^{(2)}_n=\bigl(\bar{w}_0\mathcal{L}^{n*}_{1}\bar{w}^{-1}_0\bigr)_{<0}
=\bar{w}_0(\mathcal{L}^{n*}_{1})_{<0}\bar{w}^{-1}_0
=\bar{w}_0(\mathcal{L}^{n}_{1})_{>0}^*\bar{w}^{-1}_0
=\bar{w}_0\mathcal{B}^{(1)*}_n\bar{w}^{-1}_0
-(\mathcal{L}^{n*}_{1})_{[0]}\\
\hphantom{\tilde{\mathcal{B}}^{(2)}_n}{}
=\bar{w}_0\mathcal{B}^{(1)*}_n\bar{w}^{-1}_0-\pa_{x^{(2)}_n}\bar{w}_0\cdot \bar{w}_0^{-1}.
\end{gather*}
Finally, with these two relations, we can prove this proposition by direct computation.
\end{proof}

\begin{Remark}
Note that $\pi_0^2=1$, thus we have $\pi_0^{-1}=\pi_0$.
\end{Remark}
	
\begin{Corollary}\label{Todaeigen}
Suppose $q$ and $r$ are Toda eigenfunction and adjoint eigenfunction with respect to $\mathcal{L}_{1}$, respectively, and $\tilde{\mathcal{L}_j}$ and $\tilde{x}^{(j)}_n$, $j=1,2$, are defined in Proposition~$\ref{automorToda}$, then we have
\begin{alignat*}{3}
&\pa_{\tilde{x}^{(1)}_n}(\bar{w}_0r(s))=\bigl(\tilde{\mathcal{L}}_1^n\bigr)_{\geq0}(\bar{w}_0r(s)),\qquad&&
\pa_{\tilde{x}^{(2)}_n}(\bar{w}_0r(s))=\bigl(\tilde{\mathcal{L}}_2^n\bigr)_{<0}(\bar{w}_0r(s)), \\
&\pa_{\tilde{x}^{(1)}_n}\bigl(\bar{w}_0^{-1}q(s)\bigr)=-\bigl(\bigl(\tilde{\mathcal{L}}_1^n\bigr)_{\geq0}\bigr)^*\bigl(\bar{w}_0^{-1}q(s)\bigr),\qquad&&
\pa_{\tilde{x}^{(2)}_n}\bigl(\bar{w}_0^{-1}q(s)\bigr)=-\bigl(\bigl(\tilde{\mathcal{L}}_2^n\bigr)_{<0}\bigr)^*\bigl(\bar{w}_0^{-1}q(s)\bigr),&
\end{alignat*}
which means that $\bar{w}_0r$ and $\bar{w}_0^{-1}q$ can be seen as the Toda eigenfunction and adjoint eigenfunction with respect to $\tilde{\mathcal{L}}_i$.
	\end{Corollary}

	\subsection{Basic facts about the self-transformation of mToda hierarchies}
	
	The following propositions give us a self-transformation of mToda hierarchy.
	\begin{Proposition}\label{automormToda}
		Assume $L_{j}$ and $S_{j}$, $j=1,2$, are the Lax operators and the dressing operators of mToda, respectively. Consider the following the map
\smash{$\pi_1\colon \bigl({L}_i,{S}_i,x_n^{(1)},x_n^{(2)}\bigr)\!\rightarrow\!
\bigl(\tilde{{L}}_i,\tilde{{S}}_i,\tilde{x}_n^{(1)},\tilde{x}_n^{(2)}\bigr)$} with
\begin{alignat*}{3}
&\tilde{{L}}_1=\iota_{\Lambda^{-1}}\Delta^{-1}{L}^*_{2}\Delta,\qquad&&
\tilde{{L}}_2=\iota_{\Lambda}\Delta^{-1}{L}^*_{1}\Delta,\\
&\tilde{{S}}_1=\iota_{\Lambda^{-1}}\Delta^{-1}\bigl({S}^{-1}_{2}\bigr)^*\La, \qquad&&
\tilde{{S}}_2=-\iota_{\Lambda}\Delta^{-1}\bigl({S}^{-1}_{1}\bigr)^*,\\
&\tilde{x}_n^{(1)}=-x_n^{(2)},\qquad&& \tilde{x}_n^{(2)}=-x_n^{(1)},&
\end{alignat*}
then the map $\pi_1$ is self-transformation of the mToda hierarchy, that is,
\begin{alignat*}{3}
&\pa_{\tilde{x}_n^{(1)}}\tilde{L}_j=\bigl[\tilde{B}^{(1)}_n,\tilde{L}_j\bigr], \qquad&& \pa_{\tilde{x}_n^{(2)}}\tilde{L}_j=\bigl[\tilde{B}^{(2)}_n,\tilde{L}_j\bigr],&\\			
&\pa_{\tilde{x}_n^{(1)}}{\tilde{S}}_1=\tilde{B}^{(1)}_n{\tilde{S}}_1-{\tilde{S}}_1\La^n,\qquad&& \pa_{\tilde{x}_n^{(2)}}{\tilde{S}}_1=\tilde{B}^{(2)}_n{\tilde{S}}_1,& \\ &\pa_{\tilde{x}_n^{(1)}}{\tilde{S}}_2=\tilde{B}^{(1)}_n{\tilde{S}}_2,\qquad&& \pa_{\tilde{x}_n^{(2)}}{\tilde{S}}_2=\tilde{B}^{(2)}_n{\tilde{S}}_2-{\tilde{S}}_2\La^{-n},&
\end{alignat*}
where \smash{$\tilde{B}^{(1)}_n=\bigl(\tilde{{L}}^n_1\bigr)_{\Delta,\geq1}$} and \smash{$\tilde{B}^{(2)}_n=\bigl({\tilde{{L}}^n_2}\bigr)_{\Delta^*,\geq1}$}. For symbols $A_{\Delta,\geq 1}$ and $A_{\Delta^*,\geq 1}$, they can be found in the paragraph above Lemma $\ref{fdoprop1}$.
\end{Proposition}
\begin{proof}
Firstly, by using Lemma \ref{fdoprop2}, we can get the following relations:
\begin{align*}
&\tilde{B}^{(1)}_n=\bigl(\iota_{\Lambda^{-1}}\Delta^{-1}(L_2^n)^*\Delta\bigr)_{\Delta,\ge1}=\iota_{\Lambda^{\pm1}}
\Delta^{-1}((L_2^n)_{\Delta^{*},\ge1})^*\Delta
=\iota_{\Lambda^{\pm1}}\Delta^{-1}B^{(2)*}_n\Delta,\\
&\tilde{B}^{(2)}_n=\bigl(\iota_{\Lambda}\Delta^{-1}
(L_1^n)^*\Delta\bigr)_{\Delta^*,\ge1}=
\iota_{\Lambda^{\pm1}}\Delta^{-1}((L_1^n)_{\Delta,\ge1})^*\Delta
=\iota_{\Lambda^{\pm1}}\Delta^{-1}B^{(1)*}_n\Delta.
\end{align*}
Then by direct computation, we can at last prove this proposition.
\end{proof}

\begin{Remark}\label{remark4.5}
Firstly, recall that $L_1=u_{-1}\La+\sum_{i=0}^{\infty}u_i\La^{-i}$, $ L_2=\bar{u}_{-1}\La^{-1}+\sum_{i=0}^{\infty}\bar{u}_i\La^{i}$. Then we can find that $\tilde{{L}}_1$ and $\tilde{{L}}_2$ can be expressed by coefficients of $L_1$ and $L_2$ in the following way:
\begin{align*}
&\tilde{{L}}_1=
\bar{u}_{-1}(s,\mathbf{x})\Lambda+\sum_{l=0}^\infty \Biggl(\bar{u}_l(s-l-1,\mathbf{x})+\sum_{j=-1}^{l-1}(\bar{u}_j(s-l-1,\mathbf{x})
-\bar{u}_j(s-l,\mathbf{x}))\Biggr)\Lambda^{-l},\\
&\tilde{{L}}_2=
{u}_{-1}(s-1,\mathbf{x})\Lambda^{-1}+\sum_{l=0}^\infty \Biggl({u}_l(s+l,\mathbf{x})+\sum_{j=-1}^{l-1}({u}_j(s+l,\mathbf{x})
-{u}_j(s+l-1,\mathbf{x}))\Biggr)\Lambda^{l}.
\end{align*}
Next, notice that for arbitrary $l\in\mathbb{Z}$,
\[
\operatorname{Ad}\La^l(L_i(s))=\La^l \cdot L_i(s)\cdot\La^{-l}=L_i(s+l),\qquad\operatorname{Ad}\La^l(S_i(s))=\La^l\cdot S_i(s)\cdot\La^{-l}=S_i(s+l),
\]
thus we can find $\operatorname{Ad}\La^l$ is also mToda self-transformation. Moreover, we can find that $\operatorname{Ad}\La^l$ and~$\pi_1$ are commutative and $\operatorname{Ad} \La\circ\pi^2_1=1$, then we get
\[
\pi^{-1}_1=\operatorname{Ad}\La\circ\pi_1.
\]
\end{Remark}

\subsection{Miura transformations from mToda to Toda}
There exist two kinds of Miura transformation from mToda to Toda, where the first kind is given by the following proposition.
\begin{Proposition}\label{mura1}
Given mToda Lax operators $L_i$ and wave operators $S_i$, if denote ${T}_1=c_0^{-1}(s)$ with $c_0(s)$ being the coefficient of $\La^0$ in $S_1$,
and set
\begin{align*}	
\mathcal{L}_i=T_1L_iT_1^{-1},\qquad\mathcal{S}_i={T}_1{S}_i, \qquad i=1,2,	
\end{align*}
then $\mathcal{L}_i$ are Toda Lax operators, $\mathcal{S}_i$ are
Toda wave operators, and $c_0^{-1}(s)$ is the corresponding Toda eigenfunction.
\end{Proposition}
\begin{proof}
Firstly, for convenience denote \smash{$B^{(1)}_n=\bigl(S_1\La^n S_1^{-1}\bigr)_{\Delta,\ge1}$}, \smash{$B^{(2)}_n=\bigl(S_2\La^{-n} S_2^{-1}\bigr)_{\Delta^*,\ge1}$} and
\smash{$\mathcal{B}^{(1)}_n=\bigl(\mathcal{S}_1\La^n \mathcal{S}_1^{-1}\bigr)_{\geq0}$},
\smash{$\mathcal{B}^{(2)}_n=\bigl(\mathcal{S}_2\La^{-n} \mathcal{S}_2^{-1}\bigr)_{<0}$}. Then by
\[
\bigl(f^{-1}\cdot A\cdot f\bigr)_{P}=f^{-1}\cdot A_P\cdot f \qquad \text{for} \quad P\in\{\geq k,\leq k,>k,<k,[k]\},
\]
we can find that
\begin{align*}
&B^{(1)}_n=\bigl(S_1\La^n S_1^{-1}\bigr)_{\ge0}-\bigl(S_1\La^n S_1^{-1}\bigr)_{\ge0}(1)
=c_0\cdot \mathcal{B}^{(1)}_n\cdot c_0^{-1}-c_0\cdot\mathcal{ B}^{(1)}_n\bigl(c_0^{-1}\bigr),\\
&B^{(2)}_n=\bigl(S_1\La^n S_1^{-1}\bigr)_{<0}-\bigl(S_1\La^n S_1^{-1}\bigr)_{<0}(1)
=c_0\cdot \mathcal{B}^{(2)}_n\cdot c_0^{-1}-c_0\cdot\mathcal{ B}^{(2)}_n\bigl(c_0^{-1}\bigr).
\end{align*}
Next, by comparing the coefficients of $\La^0$ of \eqref{mTodaSx1}, we can get
\begin{align*}
\pa_{x_n^{(1)}}c_0^{-1}=\mathcal{B}^{(1)}_n\bigl(c_0^{-1}\bigr),\qquad \pa_{x_n^{(2)}}c_0^{-1}=\mathcal{B}^{(2)}_n\bigl(c_0^{-1}\bigr),
\end{align*}
which imply $c_0^{-1}$ is the Toda eigenfunction corresponding to $\mathcal{S}_i$. Therefore, we can find
\begin{align*}
B^{(i)}_n=c_0\cdot \mathcal{B}^{(i)}_n\cdot c_0^{-1}-c_0\cdot\pa_{x_n^{(i)}}\bigl(c_0^{-1}\bigr),
\end{align*}
that is,
\begin{align*}
\mathcal{B}^{(i)}_n=c_0^{-1}\cdot {B}^{(i)}_n\cdot c_0-c_0^{-1}\pa_{x_n^{(i)}}(c_0),
\end{align*}
which implies this proposition.
\end{proof}

\begin{Corollary}\label{psiphi}
Assume $\Psi_i$ and $\Psi_i^*$, $i=1,2$, are mToda wave functions and adjoint wave functions defined by \eqref{mTodawave}, respectively, and let
\begin{alignat}{3}	
&\Phi_1(s,\mathbf{x},z)=c_0^{-1}(s)\Psi_1(s,\mathbf{x},z),\qquad && \Phi_2(s,\mathbf{x},z)=c_0^{-1}(s)\Psi_2(s,\mathbf{x},z), && \label{Tmuwavef} \\
&\Phi_1^*(s,\mathbf{x},z)=-c_0(s)\Delta(\Psi_1^*(s,\mathbf{x},z)),\qquad&&
\Phi_2^*(s,\mathbf{x},z)=c_0(s)\Delta(\Psi_2^*(s,\mathbf{x},z)).&\nonumber
\end{alignat}
Then $\Phi_i$ and $\Phi_i^*$, $i=1,2$, are Toda the wave functions and adjoint wave functions of Toda hierarchy.
\end{Corollary}
\begin{proof}
Assume $W_i$ to be mToda wave operators, then $\mathcal{W}_i$ are the Toda wave operators and thus $\Phi_i=c_0^{-1}\Psi_i$ are Toda wave functions by \eqref{mTodawave} and \eqref{todawave}. Further, by
\[
\mathcal{W}_1^{*-1}=-c_0\Delta\cdot\bigl(-\iota_{\Lambda^{-1}}\Delta^{-1} W_1^{*-1}\bigr)\qquad \text{and} \qquad\mathcal{W}_2^{*-1}=c_0\Delta\cdot\bigl(\iota_{\Lambda}\Delta^{-1} W_2^{*-1}\bigr),\]
 we can obtain relations between $\Phi_i^*$ and~$\Psi_i^*$.
\end{proof}

Besides the first Miura transformation $T_1$, there also exists
a second kind of Miura transformation from mToda to Toda
\[
T_2=\pi_0\circ T_1\circ\pi^{-1}_1\colon \ (L_i,S_i)\rightarrow ({\mathcal{L}}_i,{\mathcal{S}}_i), \qquad i=1,2,
\]
which is showed in the following diagram:
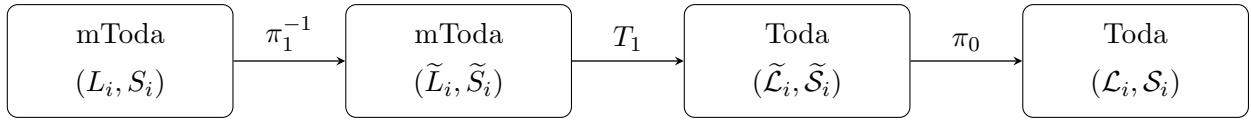
\begin{gather*}
		\tikzstyle{startstop} = [rectangle,rounded corners, minimum width=3cm,minimum height=1.5cm,text centered, draw=black,fill=white!20]
		\tikzstyle{arrow} = [thick,->,>=stealth]	
		\begin{tikzpicture}[node distance=1.95cm]
			\node (Toda0) [startstop,xshift=2cm,align=center] {mToda \\ $(L_i, S_i)$};
			\node (Toda1) [startstop,right of=Toda0,xshift=2cm,align=center] {mToda \\ $(\widetilde{L}_i,\widetilde{S}_i)$};
			\node (mToda1) [startstop,right of=Toda1,xshift=2cm,align=center] {Toda \\ $(\widetilde{\mathcal{L}_i},\widetilde{\mathcal{S}_i})$};
			\node (mToda2) [startstop,right of=mToda1,xshift=2cm,align=center] {Toda \\ $(\mathcal{L}_i,\mathcal{S}_i)$};
			\draw [arrow,line width=0.5pt] (Toda0) -- node[above]{$\pi_1^{-1}$}(Toda1);
			\draw [arrow,line width=0.5pt](Toda1)--node[above]{$T_{1}$}(mToda1);
			\draw [arrow,line width=0.5pt](mToda1)--node[above]{$\pi_0$}(mToda2);
		\end{tikzpicture}
\end{gather*}
Next, let us see the explicit procedure above.
\begin{itemize}\itemsep=0pt
\item mToda $\xrightarrow{\pi_1^{-1}}$ mToda.
By Remark~\ref{remark4.5}, $\pi^{-1}_1=\operatorname{Ad}\La\circ\pi_1$, thus we have
\begin{gather*}
{L}_1(s)\xrightarrow{\pi_1^{-1}}\widetilde{{L}}_1(s)=\iota_{\Lambda^{-1}}\Delta^{-1}{L}^*_{2}(s+1)\Delta,\\
{L}_2(s)\xrightarrow{\pi_1^{-1}}\widetilde{{L}}_2(s)=\iota_{\Lambda}\Delta^{-1}{L}^*_{1}(s+1)\Delta,\\
{S}_1(s)\xrightarrow{\pi_1^{-1}}\widetilde{{S}}_1(s)=\iota_{\Lambda^{-1}}\Delta^{-1}\bigl({S}^{-1}_{2}(s+1)\bigr)^*\Lambda, \\
{S}_1(s)\xrightarrow{\pi_1^{-1}}\widetilde{{S}}_2(s)=-\iota_{\Lambda}\Delta^{-1}\bigl({S}^{-1}_{1}(s+1)\bigr)^*.
\end{gather*}
If denote ${c}_0(s)$ and $\bar {c}_0(s)$ to be the coefficients of $\La^0$ in $S_1$ and ${S}_2$, respectively, it can be found that the coefficient of $\La^0$ in $\widetilde{{S}}_1$ is $\bar{c}^{-1}_0(s)$, and $\Lambda^0$-term in $\widetilde{{S}}_2$ is $c_0^{-1}(s+1)$.
\item mToda $\xrightarrow{T_1}$ Toda.
In this case, $T_1=\bar {c}_0(s)$,
\begin{align*}			&\widetilde{{L}}_1\xrightarrow{T_1}\widetilde{\mathcal{L}}_1
=\bar{c}_0(s)\widetilde{L}_1\bar{c}^{-1}_0(s)= \bar{c}_0(s)\iota_{\Lambda^{-1}}\Delta^{-1}{L}^*_{2}(s+1)\Delta\bar{c}^{-1}_0(s),\\
&\widetilde{{L}}_2\xrightarrow{T_1}\widetilde{\mathcal{L}}_2= \bar{c}_0(s)\widetilde{L}_2\bar{c}^{-1}_0(s)=
\bar{c}_0(s)\iota_{\Lambda}\Delta^{-1}{L}^*_{1}(s+1)\Delta\bar{c}^{-1}_0(s),\\ &\widetilde{{S}}_1\xrightarrow{T_1}\widetilde{\mathcal{S}}_1=
\bar{c}_0(s)\tilde{S}_1=\bar{c}_0(s)\iota_{\Lambda^{-1}}\Delta^{-1}\bigl({S}^{-1}_{2}(s+1)\bigr)^*\La, \\		&\widetilde{{S}}_2\xrightarrow{T_1}\widetilde{\mathcal{S}}_2=\bar{c}_0(s)\tilde{S}_2=-\bar{c}_0(s)\iota_{\Lambda}\Delta^{-1}\bigl({S}^{-1}_{1}(s+1)\bigr)^*.
\end{align*}
Now the coefficient of $\La^0$ in $\widetilde{\mathcal{S}}_2$ is $\bar{c}_0(s)c^{-1}_0(s+1)$.
\item Toda $\xrightarrow{\pi_0}$ Toda. Then
\begin{gather*}	
\widetilde{\mathcal{L}}_1\xrightarrow{\pi_0}{\mathcal{L}}_1 =\bar{c}_0(s)c^{-1}_0(s+1)\widetilde{\mathcal{L}}^*_2\bar{c}^{-1}_0(s)c_0(s+1)\\
\hphantom{\widetilde{\mathcal{L}}_1\xrightarrow{\pi_0}{\mathcal{L}}_1}{}
=c_0^{-1}(s+1)\Delta \cdot L_1(s)\cdot\iota_{\Lambda^{-1}}\Delta^{-1} c_0(s+1),\\
\widetilde{\mathcal{L}}_2\xrightarrow{\pi_0}{\mathcal{L}}_2= \bar{c}_0(s)c^{-1}_0(s+1)\widetilde{\mathcal{L}}^*_1\bar{c}^{-1}_0(s)c_0(s+1)\\
\hphantom{\widetilde{\mathcal{L}}_2\xrightarrow{\pi_0}{\mathcal{L}}_2}{}
=c_0^{-1}(s+1)\Delta \cdot L_2(s)\cdot \iota_{\Lambda}\Delta^{-1} c_0(s+1),\\
\widetilde{\mathcal{S}}_1\xrightarrow{\pi_0}{\mathcal{S}}_1=\bar{c}_0(s)c^{-1}_0(s+1)\widetilde{\mathcal{S}}_2^{-1*}
=c_0^{-1}(s+1)\Delta\cdot {S}_1(s)\Lambda^{-1}, \\
\widetilde{\mathcal{S}}_2\xrightarrow{\pi_0}{\mathcal{S}}_2=\bar{c}_0(s)c^{-1}_0(s+1)\widetilde{\mathcal{S}}_1^{-1*}=-c_0^{-1}(s+1)\Delta\cdot{S}_2(s).
\end{gather*}
\end{itemize}

Let us summarize above discussion into the following proposition.
\begin{Proposition}\label{mura2}
Suppose $L_i$ and ${S}_i$, $i=1,2$, are mToda Lax and wave operators, respectively,
where $c_0(s)$ is the coefficient of $\La^0$ in $S_1$, and define
\begin{alignat*}{3}	
&\mathcal{L}_1=T_2 L_1\iota_{\Lambda^{-1}}T_2^{-1},\qquad&&
\mathcal{L}_2=T_2 L_2\iota_{\Lambda}T_2^{-1},&\\
&\mathcal{S}_1={T}_2{S}_1\Lambda^{-1},\qquad&&
\mathcal{S}_2=-{T}_2{S}_2,&
\end{alignat*}
with operator $T_2$ given by
\begin{align*}
{T}_2=c_0^{-1}(s+1)\Delta,
\end{align*}
then $\mathcal{L}_i$ and $\mathcal{S}_i$, $i=1,2$, are Toda Lax and wave operators, respectively.
\end{Proposition}
\begin{Corollary}\label{todaadeig}
Under the same condition in Proposition $\ref{mura2}$,
$c_0(s+1)$ is the Toda adjoint eigenfunction with respect to Lax operators ${\mathcal{L}}_i$, $i=1,2$.
\end{Corollary}
\begin{proof}
In fact according to Lemma \ref{fdoprop3}, we have
\begin{align*}
\pa_{x_n^{(1)}}S_1=\left((L^n_1)_{\Delta,\geq1}-L^n_1\right){S}_1
&{}=-\iota_{\La^{-1}}\Delta^{-1} c_0(s+1)\cdot(\mathcal{L}_1(s)^n)_{<0}\cdot c_0^{-1}(s+1)\Delta\cdot S_1(s) \notag\\
&{}\quad-\iota_{\La^{-1}}\Delta^{-1}\cdot (\mathcal{L}_1(s)^{n})^*_{\geq 0}(c_0(s+1))\cdot c_0^{-1}(s+1)\Delta \cdot S_1(s).
\end{align*}	
Compare the coefficients of $\La^0$ on both sides of the above equation, we have
\begin{align*}
\pa_{x_n^{(1)}}(c_0(s))=-((\mathcal{L}_1(s-1)^n)_{\geq 0})^*(c_0(s)).
\end{align*}
Similarly, we can get $\pa_{x_n^{(2)}}(c_0(s))=-((\mathcal{L}_2(s-1)^n)_{<0})^*(c_0(s))$.
\end{proof}

\begin{Corollary}\label{psiphi-2}
Assume $\Psi_i$ and $\Psi_i^*$, $i=1,2$, are the mToda wave functions and adjoint wave functions and let
\begin{alignat*}{3}	
&\Phi_1(s,\mathbf{x},z)=z^{-1}c_0^{-1}(s+1)\Delta(\Psi_1(s,\mathbf{x},z)),\qquad&&
\Phi_2(s,\mathbf{x},z)=-c_0^{-1}(s+1)\Delta(\Psi_2(s,\mathbf{x},z)),& \\
& \Phi_1^*(s,\mathbf{x},z)=zc_0(s+1)\Psi_1^*(s+1,\mathbf{x},z), \qquad&&
\Phi_2^*(s,\mathbf{x},z)=c_0(s+1)\Psi_2^*(s+1,\mathbf{x},z),&
\end{alignat*}
then $\Phi_i$ and $\Phi_i^*$, $i=1,2$, are Toda wave functions and adjoint wave functions, respectively.
	\end{Corollary}
\subsection{Anti-Miura transformation from Toda to mToda}
There exist two kinds of anti-Miura transformation from Toda to mToda, where the first kind is given by the following proposition.
\begin{Proposition}
Given Toda wave operators $\mathcal{S}_i$, $i=1,2$, and Toda Lax operators $\mathcal{L}$, $i=1,2$, if denote $\mathcal{T}_1=q(s)^{-1}$ with $q(s)$ being Toda eigenfunction
defined by \eqref{todaeigen}, then $L_i=\mathcal{T}_1\mathcal{L}_i\mathcal{T}_1^{-1} $ and $S_i=\mathcal{T}_1\mathcal{S}_i$
will be the mToda Lax operator and mToda wave operator, respectively.
\end{Proposition}
\begin{proof}
Firstly, if for convenience denote \smash{$B^{(1)}_n=\bigl(S_1\La^n S_1^{-1}\bigr)_{\Delta,\ge1}$}, \smash{$B^{(2)}_n=\bigl(S_2\La^{-n} S_2^{-1}\bigr)_{\Delta^*,\ge1}$} and \smash{$\mathcal{B}^{(1)}_n=\bigl(\mathcal{S}_1\La^n \mathcal{S}_1^{-1}\bigr)_{\Lambda,\geq0}$}, \smash{$\mathcal{B}^{(2)}_n=\bigl(\mathcal{S}_2\La^{-n} \mathcal{S}_2^{-1}\bigr)_{\Lambda,<0}$}, then according to Lemma \ref{fdoprop3},
\begin{align*}
&B^{(1)}_n=(q^{-1}\mathcal{S}_1\La^n \mathcal{S}_1^{-1}q)_{\Delta,\ge1}=q^{-1}\mathcal{B}^{(1)}_n\cdot q-q^{-1}\mathcal{B}^{(1)}_n(q),\\
&B^{(2)}_n=(q^{-1}\mathcal{S}_2\La^n \mathcal{S}_2^{-1}q)_{\Delta^*,\ge1}=q^{-1}\mathcal{B}^{(2)}_n\cdot q-q^{-1}\mathcal{B}^{(2)}_n(q).
\end{align*}
The next computations are quite direct from the corresponding evolutions of $\mathcal{L}_i$ and $\mathcal{S}_i$.
\end{proof}

\begin{Corollary}\label{T1phipsi}
Assume $\Phi_i$ and $\Phi_i^*$, $i=1,2$, to be Toda wave functions and adjoint wave functions defined by \eqref{todawave}, respectively, and let
\begin{alignat*}{3}
&\Psi_1(s,\mathbf{x},z)=q^{-1}(s)\Phi_1(s,\mathbf{x},z),\qquad&&
\Psi_2(s,\mathbf{x},z)=q^{-1}(s)\Phi_2(s,\mathbf{x},z),&	 \\
&\Psi_1^*(s,\mathbf{x},z)=-\iota_{\La}\Delta^{-1}q(s)(\Phi_1^*(s,\mathbf{x},z)),\qquad&&
\Psi_2^*(s,\mathbf{x},z)=\iota_{\La^{-1}}\Delta^{-1}q(s)(\Phi_2^*(s,\mathbf{x},z)),&
\end{alignat*}
then $\Psi_i$ and $\Psi_i^*$, $i=1,2$, are mToda wave functions and adjoint wave functions.
\end{Corollary}

Similar to the case of Miura transformation, we can also define
the second kind of anti-Miura transformation from Toda to mToda constructed by
\begin{gather*}
\mathcal{T}_2=\pi_1\circ \mathcal{T}_1\circ\pi_0\colon \ ({\mathcal{L}}_i,{\mathcal{S}}_i)\rightarrow (L_i,S_i),\qquad i=1,2,
\\
		\tikzstyle{startstop} = [rectangle,rounded corners, minimum width=3cm,minimum height=1.5cm,text centered, draw=black,fill=white!20]
		\tikzstyle{arrow} = [thick,->,>=stealth]	
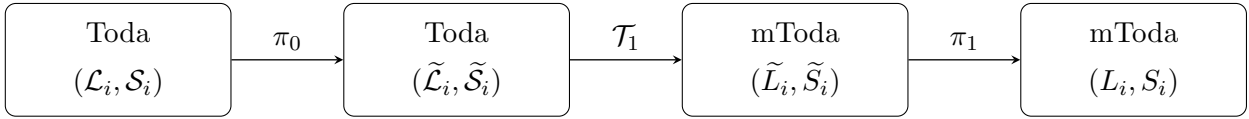
		\begin{tikzpicture}[node distance=1.9cm]
			\node (Toda0) [startstop,xshift=2cm,align=center] {Toda \\ $(\mathcal{L}_i, \mathcal{S}_i)$};
			\node (Toda1) [startstop,right of=Toda0,xshift=2cm,align=center] {Toda \\ $(\widetilde{\mathcal{L}}_i,\widetilde{\mathcal{S}}_i)$};
			\node (mToda1) [startstop,right of=Toda1,xshift=2cm,align=center] {mToda \\ $(\widetilde{L_i},\widetilde{S_i})$};
			\node (mToda2) [startstop,right of=mToda1,xshift=2cm,align=center] {mToda \\ $(L_i,S_i)$};
			\draw [arrow,line width=0.5pt] (Toda0) -- node[above]{$\pi_0$}(Toda1);
			\draw [arrow,line width=0.5pt](Toda1)--node[above]{$\mathcal{T}_{1}$}(mToda1);
			\draw [arrow,line width=0.5pt](mToda1)--node[above]{$\pi_1$}(mToda2);
\end{tikzpicture}
\end{gather*}
\begin{itemize}\itemsep=0pt
\item Toda $\xrightarrow{\pi_0}$ Toda,
\begin{align*}			\mathcal{L}_i\xrightarrow{\pi_0}\widetilde{\mathcal{L}}_i=\bar{w}_0\mathcal{L}^*_{3-i}\bar{w}^{-1}_0,\qquad \mathcal{S}_i\xrightarrow{\pi_0} \widetilde{\mathcal{S}}_i=\bar{w}_0(\mathcal{S}^{-1}_{3-i})^*,
\end{align*}
where $\bar{w}_0$ is $\La^0$-coefficient in $\mathcal{S}_2$.
By using Corollary \ref{Todaeigen}, $\bar{w}_0(s)r(s)$ is the Toda eigenfunction with respect to $\widetilde{\mathcal{L}}_i$ for Toda adjoint eigenfunction $r(s)$
corresponding to Lax operator~$\mathcal{L}_i$.		
\item Toda $\xrightarrow{\mathcal{T}_1}$ mToda,
\begin{align*}			&\widetilde{\mathcal{L}}_i\xrightarrow{\mathcal{T}_1}\widetilde{L}_i=(\bar{w}_0r(s))^{-1}\tilde{\mathcal{L}}_i(s)(\bar{w}_0r(s))=r^{-1}(k)\mathcal{L}(s)^*_{3-i}r(s), \\ &\widetilde{\mathcal{S}}_i\xrightarrow{\mathcal{T}_1}\widetilde{S}_i=(\bar{w}_0r(s))^{-1}\tilde{\mathcal{S}}_i(s)=r^{-1}(s)(\mathcal{S}^{-1}_{3-i})(s)^*.
\end{align*}
\item mToda $\xrightarrow{\pi_1}$ mToda,
\begin{align*}		
&\widetilde{L}_1\xrightarrow{\pi_1} L_1=\iota_{\Lambda^{-1}}\Delta^{-1}{\tilde{L}^*_{2}}\Delta
=\iota_{\Lambda^{-1}}\Delta^{-1}r(s)\cdot{\mathcal{L}_1(s)}\cdot r^{-1}(s)\Delta, \notag \\
&\widetilde{L}_2\xrightarrow{\pi_1} L_2=\iota_{\Lambda}\Delta^{-1}{\tilde{L}^*_{1}}\Delta=\iota_{\Lambda}\Delta^{-1}r(s)\cdot{\mathcal{L}_2(s)}\cdot r^{-1}(s)\Delta, \\
&\widetilde{S}_1\xrightarrow{\pi_1} S_1=\iota_{\Lambda^{-1}}\Delta^{-1}\bigl(\tilde{{S}}^{-1}_{2}\bigr)^*\La
=\iota_{\Lambda^{-1}}\Delta^{-1}r(s)\cdot\mathcal{S}_1(s)\Lambda, \\
&\widetilde{S}_2\xrightarrow{\pi_1} S_2=-\iota_{\Lambda}\Delta^{-1}\bigl(\tilde{S}^{-1}_{1}\bigr)^*
=-\iota_{\Lambda}\Delta^{-1}r(s)\cdot\mathcal{S}_2(s).
\end{align*}
\end{itemize}

Let us summarize above discussion into the following proposition.
\begin{Proposition}\label{T2:to-mto}
Given Toda Lax operators $\mathcal{L}_i$, Toda wave operators $\mathcal{W}_i$, 	Toda wave functions~$\Phi_i$, $i=1,2$, Toda adjoint wave functions $\Psi_i^*$ and Toda adjoint eigenfunction $r(s)$,
if define%
\begin{alignat*}{3}
&\mathcal{T}_2=\Delta^{-1}r(s),&\\
&L_1=\iota_{\Lambda^{-1}}\mathcal{T}_2 \cdot\mathcal{L}_1\cdot\mathcal{T}_2^{-1},\qquad&&
L_2=\iota_{\Lambda}T_n \cdot \mathcal{L}_2\cdot\mathcal{T}_2^{-1},&\notag \\
&{S}_1=\iota_{\Lambda^{-1}}\mathcal{T}_2\cdot\mathcal{S}_1\Lambda,\qquad&&
 {S}_2=-\iota_{\Lambda}\mathcal{T}_2\cdot\mathcal{S}_2,&
\end{alignat*}
and let
\begin{alignat*}{3}
&\Psi_1(s,\mathbf{x},z)=z\iota_{\Lambda^{-1}}\Delta^{-1}r(s)(\Phi_1(s,\mathbf{x},z)),\qquad &&
\Psi_2(s,\mathbf{x},z)=-\iota_{\Lambda}\Delta^{-1}r(s)(\Phi_2(s,\mathbf{x},z)),& \\
&\Psi_1^*(s,\mathbf{x},z)=z^{-1}r^{-1}(s-1)\Phi_1^*(s-1,\mathbf{x},z),\qquad&&
\Psi_2^*(s,\mathbf{x},z)=r^{-1}(s-1)\Phi_2^*(s-1,\mathbf{x},z),&
\end{alignat*}
then $L_i$, $S_i$, $\Psi_i$ and $\Psi_i^*$ are corresponding mToda objects.
\end{Proposition}

\section{mToda tau functions}\label{section5}
In this section, we will discuss mToda tau functions and their relations with Toda tau functions.
Note that from mToda Lax formulation, we have proved the mToda bilinear equation in terms of wave functions. Firstly, let us prove the existence of tau functions for mToda hierarchy from mToda bilinear equation.
\begin{Proposition}
Given mToda wave functions $\Psi_i$ and mToda adjoint wave functions $\Psi_i^*$, $i=1,2$, satisfying
\begin{eqnarray}\label{2mkpbilinearwave-2}
\oint_{C_\infty}\frac{{\rm d}z}{2\pi {\rm i}z }\Psi_1(s,\mathbf{x},z)\Psi_1^*\bigl(s',\mathbf{x}',z\bigr)+\oint_{C_0}\frac{{\rm d}z}{2\pi {\rm i}z }\Psi_2(s,\mathbf{x},z)\Psi_2^*\bigl(s',\mathbf{x}',z\bigr)=1,
\end{eqnarray}
then there exist tau functions $\tau_{0,s}(\mathbf{x})$ and $\tau_{1,s}(\mathbf{x})$ such that
\begin{align*}
&\Psi_1(s,\mathbf{x},z)=\frac{\tau_{0,s}\bigl(\mathbf{x}-\bigl[z^{-1}\bigr]_1\bigr)}{\tau_{1,s}(\mathbf{x})}z^{s}{\rm e}^{\xi(\mathbf{x}^{(1)},z)},\\ 	&\Psi_1^*(s,\mathbf{x},z)=\frac{\tau_{1,s}\bigl(\mathbf{x}+\bigl[z^{-1}\bigr]_1\bigr)}{\tau_{0,s}(\mathbf{x})}z^{-s}{\rm e}^{-\xi(\mathbf{x}^{(1)},z)},\\ &\Psi_2(s,\mathbf{x},z)=\frac{\tau_{0,s+1}(\mathbf{x}-[z]_2)}{\tau_{1,s}(\mathbf{x})}z^{s}{\rm e}^{\xi(\mathbf{x}^{(2)},z^{-1})},\\ &\Psi_2^*(s,\mathbf{x},z)=\frac{\tau_{1,s-1}(\mathbf{x}+[z]_2)}{\tau_{0,s}(\mathbf{x})}z^{-s+1}{\rm e}^{-\xi(\mathbf{x}^{(2)},z^{-1})},
\end{align*}
where $(\tau_{0,s}(\mathbf{x}),\tau_{1,s}(\mathbf{x}))$ is called mToda tau pair.
\end{Proposition}
\begin{proof}
Firstly, by Corollary \ref{psiphi}, we can find that $\Phi_j(s,\mathbf{x},z)$, $j=1,2$, defined by (\ref{Tmuwavef}) are Toda wave functions, thus there exists
tau functions $\tau_s^{\rm Toda}$ such that
\begin{align*}
&c_0^{-1}(s)\Psi_1(s,\mathbf{x},z)=\frac{\tau^{{\rm Toda}}_s\bigl(\mathbf{x}-\bigl[z^{-1}\bigr]_1\bigr)}
{\tau^{{\rm Toda}}_s(\mathbf{x})}{\rm e}^{\xi(\mathbf{x}^{(1)},z)}z^s, \\
&c_0^{-1}(s)\Psi_2(s,\mathbf{x},z)=\frac{\tau^{{\rm Toda}}_{s+1}(\mathbf{x}-[z]_2)}
{\tau^{{\rm Toda}}_s(\mathbf{x})}{\rm e}^{\xi(\mathbf{x}^{(2)},z^{-1})}z^s.
\end{align*}
So, if denote \smash{$\tau_{0,s}(\mathbf{x})=\tau^{{\rm Toda}}_s(x)$}, \smash{$\tau_{1,s}(\mathbf{x})=c_0^{-1}(s)\tau^{{\rm Toda}}_s(x)$}, then we can get relations for $\Psi_1$ and~$\Psi_2$.

Next, if let
\smash{$\mathbf{x}'=\mathbf{x}+\bigl[\lambda^{-1}\bigr]_1$} and $s'=s$ in \eqref{2mkpbilinearwave-2}, we can get
\begin{align*}
w^{(1)}\bigl(s,\mathbf{x}+\bigl[\lambda^{-1}\bigr]_1,\lambda\bigr)w^{(1)*}(s,\mathbf{x},\lambda)=1,
\end{align*}
then we can derive the expression for $\Psi_1^*(s,\mathbf{x},z)$.
Similarly, letting $\mathbf{x}'=\mathbf{x}+[\lambda]_2$, we can get
\begin{align*}
\lambda^{-1}w^{(2)}(s,\mathbf{x}+[\la]_2,\lambda)w^{(2)*}(s+1,\mathbf{x},\la)=1,
\end{align*}
from which we can get the expression for $\Psi_2^*(s,\mathbf{x},z)$.
\end{proof}

\begin{Lemma}\label{fconst}
For the function $f(s,\mathbf{x})\in\mathcal{A}$, if it satisfies
\begin{align}\label{fconstf}
{\rm e}^{\pm\xi(\tilde{\pa}_{\mathbf{x}^{(1)}},z^{-1})}f(s,\mathbf{x})=
\La^{\mp} {\rm e}^{\pm\xi(\tilde{\pa}_{\mathbf{x}^{(2)}},z)}f(s,\mathbf{x}),
\end{align}
then $f $ is a constant independent of $s$ and $\mathbf{x}$.
\end{Lemma}
\begin{proof}
Compare the coefficients of $z^j$, $j\in\mathbb{Z}$, in \eqref{fconstf}, we can get
\begin{align*}
f(s,\mathbf{x})=f(s\mp1,\mathbf{x}),\qquad \partial_{x_n^{(1)}}f(s,x)=\partial_{x_n^{(2)}}f(s\mp1,\mathbf{x})=0,
\end{align*}
which implies that $f $ is a constant independent of $s$ and $\mathbf{x}$.
\end{proof}

\begin{Proposition}\label{prop:mtodatau}
Given mToda tau pair $(\tau_{0,s}(\mathbf{x}),\tau_{1,s}(\mathbf{x}))$,
both $\tau_{0,s}(\mathbf{x})$ and $\tau_{1,s}(\mathbf{x})$ are Toda tau functions. And if denote $q=\tau_1/\tau_0$ and $r=\La(\tau_0/\tau_1)$, then $q$ is the Toda eigenfunction with respect to $\tau_0$ and $r$ is the Toda adjoint eigenfunction with respect to $\tau_1$.
\end{Proposition}
\begin{proof}
Firstly, notice that $c_0=\tau_0/\tau_1$, then
by using \eqref{TLtau} and \eqref{todaawavetau} and relations of $\Psi_i$ and $\Phi_i$ in Corollary \ref{psiphi}, we can get
\begin{align*}
\frac{\tau_s^{\rm Toda}\bigl(\mathbf{x}-\bigl[z^{-1}\bigr]_1\bigr)}{\tau_s^{\rm Toda}(\mathbf{x})}=\frac{\tau_{0,s}\bigl(\mathbf{x}-\bigl[z^{-1}\bigr]_1\bigr)}{\tau_{0,s}(\mathbf{x})},\qquad \frac{\tau_{s+1}^{\rm Toda}(\mathbf{x}-[z]_2)}{\tau_s^{\rm Toda}(\mathbf{x})}=\frac{\tau_{0,s+1}(\mathbf{x}-[z]_2)}{\tau_{0,s}(\mathbf{x})},
\end{align*}
which implies
\begin{align*}
{\rm e}^{-\xi(\tilde{\pa}_{x^{(1)}},z^{-1})}\bigl({\rm log}\tau_s^{\rm Toda}(\mathbf{x})-{\rm log}\tau_{0,s}(\mathbf{x})\bigr)= \Lambda {\rm e}^{-\xi(\tilde{\pa}_{x^{(2)}},z)}\bigl({\rm log}\tau_s^{\rm Toda}(\mathbf{x})-{\rm log}\tau_{0,s}(\mathbf{x})\bigr).
\end{align*}
Thus by Lemma \ref{fconst}, we can obtain $\tau_{0,s}={\rm const}\cdot\tau_s^{\rm Toda}$. Similarly, by relations between $\Psi_i^*$ and $\Phi_i^*$ in Corollary \ref{psiphi-2}, we can prove $\tau_{1,s}={\rm const}\cdot\tau_s^{\rm Toda}$. Finally, by Proposition \ref{mura1} and Corollary \ref{todaadeig}, we can find $q$ is the Toda eigenfunction with respect to $\tau_0$ and $r$ is the Toda adjoint eigenfunction with respect to $\tau_1$.
\end{proof}
\begin{Corollary}\label{corollary:taurelation}
From mToda to Toda, we can find for mToda tau pair $(\tau_0,\tau_1)$,
\begin{itemize}\itemsep=0pt
\item Case $T_1=c^{-1}_0(s)$,
\[
\tau_s^{\rm Toda}(\mathbf{x})=\tau_{0,s}(\mathbf{x}),
\]
\item Case $T_2=c^{-1}_0(s+1)\Delta$,
\[
\tau_s^{\rm Toda}(\mathbf{x})=\tau_{1,s}(\mathbf{x}).
\]
\end{itemize}
While from Toda to mToda,
\begin{itemize}\itemsep=0pt
\item Case $\mathcal{T}_1=q(s)^{-1}$,
 \[
 \tau_{0,s}(\mathbf{x})=\tau_s^{\rm Toda}(\mathbf{x}),\qquad \tau_{1,s}(\mathbf{x})=q(s)\tau_s^{\rm Toda}(\mathbf{x}).
 \]
 \item Case $\mathcal{T}_2=\Delta^{-1}r(s)$,
 \[
 \tau_{0,s}(\mathbf{x})=r(s-1)\tau_s^{\rm Toda}(\mathbf{x}),\qquad \tau_{1,s}(\mathbf{x})=\tau_s^{\rm Toda}(\mathbf{x}).
 \]
\end{itemize}
Here $q(s)$ and $r(s)$ are the Toda eigenfunction and adjoint eigenfunction with respect to $\tau_{s}^{\rm Toda}$, respectively.
\end{Corollary}
\begin{proof}
For the results from mToda to Toda, they can be found in the proof of Proposition \ref{prop:mtodatau}. While for the results in $\mathcal{T}_1$ from Toda to mToda, we can find from results for wave functions in Corollary \ref{T1phipsi} that
\begin{align}
\frac{\tau_s^{\rm Toda}\bigl(\mathbf{x}-\bigl[z^{-1}\bigr]_1\bigr)}{q(s,\mathbf{x})\tau_s^{\rm Toda}(\mathbf{x})}=\frac{\tau_{0,s}\bigl(\mathbf{x}-\bigl[z^{-1}\bigr]_1\bigr)}{\tau_{1,s}(\mathbf{x})},\qquad \frac{\tau_{s+1}^{\rm Toda}(\mathbf{x}-[z]_2)}{q(s,\mathbf{x})\tau_s^{\rm Toda}(\mathbf{x})}=\frac{\tau_{0,s+1}(\mathbf{x}-[z]_2)}{\tau_{1,s}(\mathbf{x})},
\label{psiphirelation}
\end{align}
which implies that
\begin{align*}
{\rm e}^{-\xi(\tilde{\pa}_{x^{(1)}},z^{-1})}\log\frac{\tau_s^{\rm Toda}(\mathbf{x})}{\tau_{0,s}(\mathbf{x})}
={\rm e}^{-\xi(\tilde{\pa}_{x^{(2)}},z)}\log\frac{\tau_{s+1}^{\rm Toda}(\mathbf{x})}{\tau_{0,s+1}(\mathbf{x})}.\label{shift-psiphirelation}
\end{align*}
By Lemma \ref{fconst}, we can obtain
$\tau_{0,s}(\mathbf{x})={\rm const}\cdot\tau_s^{\rm Toda}(\mathbf{x})$. Then we can get $\tau_{1,s}(\mathbf{x})={\rm const}\cdot q(s,\mathbf{x})\tau_s^{\rm Toda}(\mathbf{x})$ by \eqref{psiphirelation}. Similarly, by results for adjoint wave functions in Proposition \ref{T2:to-mto}, we can prove the case for $\mathcal{T}_2$.
\end{proof}

\section{Darboux transformations for Toda and mToda hierarchies}\label{section6}
Now we can use Miura and anti-Miura transformations shown in above section to construct the corresponding Darboux transformations of Toda and mToda hierarchies by the following way:
\begin{align*}
{\rm Toda} \xrightarrow{\text{anti-Miura}} {\rm mToda} \xrightarrow{{\rm Miura}} {\rm Toda},\qquad {\rm mToda} \xrightarrow{{\rm Miura}} {\rm Toda} \xrightarrow{\text{anti-Miura}} {\rm mToda}.
\end{align*}
\begin{Proposition}\label{prop:todadarboux}
Given Toda eigenfunction $q$ and Toda adjoint eigenfunction $r$, Toda Lax operators $(\mathcal{L}_1,\mathcal{L}_2)$ and wave operators $(\mathcal{S}_1,\mathcal{S}_2)$, if denote \smash{$\bigl(\mathcal{L}_1^{[1]},\mathcal{L}_2^{[1]}\bigr)$} and \smash{$\bigl(\mathcal{S}_1^{[1]},\mathcal{S}_2^{[1]}\bigr)$} in the following way:
\begin{itemize}\itemsep=0pt
\item Case $\mathcal{T}_{1,2}(q)=q(s+1)\cdot\Delta\cdot q(s)^{-1}$,
\begin{alignat*}{3}
&\mathcal{L}_1^{[1]}=\mathcal{T}_{1,2}\cdot\mathcal{L}_1\cdot\iota_{\Lambda^{-1}}
\mathcal{T}_{1,2}^{-1},\qquad&& \mathcal{L}_2^{[1]}=\mathcal{T}_{1,2}\cdot\mathcal{L}_2\cdot\iota_{\Lambda}\mathcal{T}_{1,2}^{-1},&\\
&\mathcal{S}_1^{[1]}=\mathcal{T}_{1,2}\cdot\mathcal{S}_1\cdot\Lambda^{-1},\qquad&& \mathcal{S}_2^{[1]}=-\mathcal{T}_{1,2}\cdot\mathcal{S}_2,&
\end{alignat*}
\item Case $\mathcal{T}_{2,1}(r)=r^{-1}(s-1)\cdot\Delta^{-1}\cdot r(s)$,
\begin{alignat*}{3}
&\mathcal{L}_1^{[1]}=\iota_{\Lambda^{-1}}\mathcal{T}_{2,1}\cdot\mathcal{L}_1\cdot
\mathcal{T}_{2,1}^{-1},\qquad&& \mathcal{L}_2^{[1]}=\iota_{\Lambda}\mathcal{T}_{2,1}\cdot\mathcal{L}_2\cdot
\mathcal{T}_{2,1}^{-1},&\\
&\mathcal{S}_1^{[1]}=\iota_{\Lambda^{-1}}\mathcal{T}_{2,1}\cdot\mathcal{S}_1\cdot\Lambda,\qquad&& \mathcal{S}_2^{[1]}=-\iota_{\Lambda}\mathcal{T}_{2,1}\cdot\mathcal{S}_2,&
\end{alignat*}
\end{itemize}
then \smash{$\bigl(\mathcal{L}_1^{[1]},\mathcal{L}_2^{[1]}\bigr)$} and \smash{$\bigl(\mathcal{S}_1^{[1]},\mathcal{S}_2^{[1]}\bigr)$} are new Toda Lax operators and new Toda wave operators, respectively.
\end{Proposition}
\begin{proof}
Firstly, by Proposition \ref{T2:to-mto}, under the action of $\mathcal{T}_2$,
\begin{align*}
\mathcal{S}_1\xrightarrow{\mathcal{T}_2=\Delta^{-1}r(s)}{S}_1=\iota_{\Lambda^{-1}}\Delta^{-1}r(s)\mathcal{S}_1\La,
\end{align*}
where the coefficient of $\La^0$ in $S_1$ is $r(s-1)$,
then under the action of $T_1=r^{-1}(s-1)$,
\begin{align*}
S_1\xrightarrow{T_{1}=r^{-1}(s-1)}\mathcal{S}^{[1]}_1=r^{-1}(s-1)S_1=r^{-1}(s-1) \iota_{\Lambda^{-1}}\Delta^{-1}r(s)\mathcal{S}_1\La.
\end{align*}
Similarly, we can prove
\[
\mathcal{S}^{[1]}_2=-r^{-1}(s-1)\iota_{\Lambda}\Delta^{-1}r(s)\mathcal{S}_2.
\]
 The form of $\mathcal{L}^{[1]}_i$, $i=1,2$, are given by \smash{$\mathcal{L}^{[1]}_1=\mathcal{S}^{[1]}_1\Lambda\mathcal{S}^{[1]-1}_1$} and \smash{$\mathcal{L}^{[1]}_2=\mathcal{S}^{[1]}_2\Lambda^{-1}\mathcal{S}^{[1]-1}_2$}.
As for case $\mathcal{T}_{1,2}(q)=q(s+1)\cdot\Delta\cdot q(s)^{-1}$, it can be similarly proved.
\end{proof}

\begin{Remark} The following transformations are trivial:
\[
{\rm Toda} \xrightarrow{\mathcal{T}_i} {\rm mToda} \xrightarrow{T_i} {\rm Toda},\qquad i=1,2.
\]
\end{Remark}

\begin{Corollary}
Under the same conditions of Proposition $\ref{prop:todadarboux}$, if assume that $\tau_s^{\rm Toda}$ is the Toda tau function corresponding to $(\mathcal{L}_1,\mathcal{L}_2)$, and $\Psi_i$, $i=1,2$, is the Toda wave functions, then~\smash{$\tau_s^{\rm Toda,[1]}$} given below is the new Toda tau functions corresponding to \smash{$\bigl(\mathcal{L}_1^{[1]},\mathcal{L}_2^{[1]}\bigr)$}.
\begin{itemize}\itemsep=0pt
\item Case $\mathcal{T}_{1,2}(q)=q(s+1)\cdot\Delta\cdot q(s)^{-1}$,
\begin{align*}
\tau_s^{\rm Toda,[1]}=q(s)\tau_s^{\rm Toda}.
\end{align*}
\item Case $\mathcal{T}_{2,1}(r)=r^{-1}(s-1)\cdot\Delta^{-1}\cdot r(s)$,
\begin{align*}
\tau_s^{\rm Toda,[1]}=r(s-1)\tau_s^{\rm Toda}.
\end{align*}
\end{itemize}
\end{Corollary}

\begin{proof}
Since Toda Darboux transformation $\mathcal{T}_{1,2}(q)=q(s+1)\cdot\Delta\cdot q(s)^{-1}$ is determined by
\[
{\rm Toda} \xrightarrow{\mathcal{T}_1} {\rm mToda} \xrightarrow{T_2} {\rm Toda},
\]
thus by Corollary \ref{corollary:taurelation}, we can find that
\[
\tau_s^{\rm Toda} \xrightarrow{\mathcal{T}_1}(\tau_0(s),\tau_1(s))=\bigl(\tau_s^{\rm Toda},q(s)\tau_s^{\rm Toda}\bigr) \xrightarrow{T_2} \tau_s^{\rm Toda,[1]}=q(s)\tau_s^{\rm Toda}.
\]
Similarly, we can obtain the result in case $\mathcal{T}_{2,1}(r)=r^{-1}(s-1)\cdot\Delta^{-1}\cdot r(s)$.
\end{proof}

\begin{Proposition}
For mToda hierarchy, given wave operators $(S_1,S_2)$, Lax operators $(L_1,L_2)$, tau pair $(\tau_0,\tau_1)$, eigenfunction $\Theta(s)$ and adjoint eigenfunction $\bar{\Theta}(s)$, if denote $A^{[1]}$ to be the transformed object $A$ by the following way:
\begin{itemize}\itemsep=0pt
\item Case $T_{1,1}(\Theta)=\Theta(s)^{-1}$,
\begin{alignat*}{3}
&L_1^{[1]}=T_{1,1}\cdot L_1\cdot T_{1,1}^{-1},\qquad&& L_2^{[1]}=T_{1,1}\cdot L_2\cdot T_{1,1}^{-1},&\\
&S_1^{[1]}=T_{1,1}\cdot S_1,\qquad&& S_2^{[1]}=T_{1,1}\cdot S_2,&\\
&\tau^{[1]}_{0,s}(\mathbf{x})=\tau_{0,s}(\mathbf{x}),\qquad &&\tau^{[1]}_{1,s}(\mathbf{x})=\Theta(s)\cdot\tau_{1,s}(\mathbf{x}),&
\end{alignat*}
\item Case $T_{1,2}(\bar{\Theta})=\Delta^{-1}\cdot\Delta\bigl(\bar{\Theta}(s)\bigr)$,
\begin{alignat*}{3}
&L_1^{[1]}=\iota_{\Lambda^{-1}}T_{1,2}\cdot L_1\cdot T_{1,2}^{-1},\qquad&& L_2^{[1]}=\iota_{\Lambda}T_{1,2}\cdot L_2\cdot T_{1,2}^{-1},&\\
&S_1^{[1]}=\iota_{\Lambda^{-1}}T_{1,2}\cdot S_1\cdot \La,\qquad&& S_2^{[1]}=-\iota_{\Lambda}T_{1,2}\cdot S_2,&\\ &\tau^{[1]}_{0,s}=\tau_{0,s}(\mathbf{x})\tau_{0,s-1}(\mathbf{x})\Delta\bigl(\bar{\Theta}(s-1)\bigr)/\tau_{1,s-1}(\mathbf{x}),\qquad&& \tau^{[1]}_{1,s}(\mathbf{x})=\tau_{0,s}(\mathbf{x}),&
\end{alignat*}
\item Case $T_{2,1}(\Theta)=(\Delta(\Theta(s)))^{-1}\cdot\Delta $,
\begin{alignat*}{3}
&L_1^{[1]}=T_{2,1}\cdot L_1\cdot \iota_{\Lambda^{-1}}T_{2,1}^{-1},\qquad&& L_2^{[1]}=T_{2,1}\cdot L_2\cdot \iota_{\Lambda} T_{2,1}^{-1},&\\
&S_1^{[1]}=T_{2,1}\cdot S_1\cdot \La^{-1},\qquad && S_2^{[1]}=-T_{2,1}\cdot S_2,&\\
&\tau^{[1]}_{0,s}=\tau_{1,s}, \qquad && \tau^{[1]}_{1,s}=\tau_{1,s}(\mathbf{x})\tau_{1,s+1}(\mathbf{x})
\Delta(\Theta(s))/\tau_{0,s+1}(\mathbf{x}),&
\end{alignat*}
\item Case $T_{2,2}(\bar{\Theta})=\Delta^{-1}\cdot\bar{\Theta}(s+1)\cdot\Delta$,
\begin{alignat*}{3}
&L_1^{[1]}=\iota_{\Lambda^{-1}}T_{2,2}\cdot L_1\cdot T_{2,2}^{-1},\qquad&& L_2^{[1]}= \iota_{\Lambda}T_{2,2}\cdot L_2\cdot T_{2,2}^{-1},&\\
&S_1^{[1]}=\iota_{\Lambda^{-1}}T_{2,2}\cdot S_1,\qquad&& S_2^{[1]}=\iota_{\Lambda}T_{2,2}\cdot S_2,&\\
&\tau^{[1]}_{0,s}=\bar{\Theta}(s)\cdot\tau_{0,s}, \qquad&& \tau^{[1]}_{1,s}=\tau_{1,s}(\mathbf{x}),&
\end{alignat*}
\end{itemize}
then \smash{$\bigl(S_1^{[1]},S_2^{[1]}\bigr)$}, \smash{$\bigl(L_1^{[1]},L_2^{[1]}\bigr)$} and \smash{$\bigl(\tau_0^{[1]},\tau_1^{[1]}\bigr)$} are mToda wave operators, Lax operators and tau pair, respectively.
\end{Proposition}
\begin{proof}
Here we only give the proof for the case of $T_{2,1}$ since other cases can be proved similarly. Firstly, by Proposition \ref{mura2}, under the action of $T_{2}$,
\begin{align*}
{S}_1\xrightarrow{T_{2}=c_0^{-1}(s+1)\Delta }\mathcal{S}_1=c_0^{-1}(s+1)\Delta \cdot {S}_1\cdot \La^{-1},
\end{align*}
where $c_0(s)$ is the coefficient of $\La^0$ in $S_1$. Notice that mToda eigenfunction $\Theta$ will become Toda eigenfunction $q(s)=c_0^{-1}(s+1)\Delta(\Theta(s))$, then under the action of $\mathcal{T}_{1}$
\begin{align*}			\mathcal{S}_1\xrightarrow{\mathcal{T}_{1}=q^{-1}(s)}S^{[1]}_1=q^{-1}(s)\mathcal{S}_1
=(\Delta(\Theta))^{-1}\Delta \cdot S_1\cdot \La^{-1}.
\end{align*}
In the same way, we can get $S^{[1]}_2=-(\Delta(\Theta))^{-1}\Delta S_2$.
The form of $L^{[1]}_i$, $i=1,2$, are given by \smash{$L^{[1]}_i=S^{[1]}_i\Lambda^{3-2i} S^{[1]-1}_i$}. As for the transformations of tau functions, we can find them from the following transformation derived by Corollary \ref{corollary:taurelation}:
\begin{gather*}
(\tau_{0,s}(\mathbf{x}),\tau_{1,s}(\mathbf{x}))\xrightarrow{{T}_{2}}\tau_s^{\rm Toda}(\mathbf{x})=\tau_{1,s}(\mathbf{x})
\ {\xrightarrow{\mathcal{T}_{1}}}\\
\bigl(\tau^{[1]}_{0,s}(\mathbf{x}),\tau^{[1]}_{1,s}(\mathbf{x})\bigr)=\bigl(\tau_s^{\rm Toda}(\mathbf{x}),q(s)\tau_s^{\rm Toda}(\mathbf{x})\bigr)\\
\hphantom{\bigl(\tau^{[1]}_{0,s}(\mathbf{x}),\tau^{[1]}_{1,s}(\mathbf{x})\bigr)}{}
=(\tau_{1,s}(\mathbf{x}),\tau_{1,s}(\mathbf{x})\tau_{1,s+1}(\mathbf{x})
\Delta(\Theta(s))/\tau_{0,s+1}(\mathbf{x})).
\tag*{\qed}
\end{gather*}
\renewcommand{\qed}{}
\end{proof}

\subsection*{Acknowledgements}
 Many thanks to the anonymous referees for their useful comments and suggestions. This work is supported by National Natural Science Foundations of China (Nos. 12171472 and 12261072) and ``Qinglan Project'' of Jiangsu Universities.

\pdfbookmark[1]{References}{ref}
\LastPageEnding


\begin{thebibliography}{99}
\footnotesize\itemsep=0pt

\bibitem{Adler1999}
Adler M., van Moerbeke P., Vertex operator solutions to the discrete
 {KP}-hierarchy, \href{https://doi.org/10.1007/s002200050609}{\textit{Comm. Math. Phys.}} \textbf{203} (1999), 185--210,
 \href{https://arxiv.org/abs/solv-int/9912014}{arXiv:solv-int/9912014}.

\bibitem{cheng2018jgp}
Cheng J., Li M., Tian K., On the modified {KP} hierarchy: tau functions,
 squared eigenfunction symmetries and additional symmetries, \href{https://doi.org/10.1016/j.geomphys.2018.07.022}{\textit{J.~Geom.
 Phys.}} \textbf{134} (2018), 19--37.

\bibitem{daihh}
Dai H.H., Geng X., Explicit solutions of the {$2+1$}-dimensional modified
 {T}oda lattice through straightening out of the relativistic {T}oda flows,
 \href{https://doi.org/10.1143/JPSJ.72.3063}{\textit{J.~Phys. Soc. Japan}} \textbf{72} (2003), 3063--3069.

\bibitem{Dickey1999lmp}
Dickey L.A., Modified {KP} and discrete {KP}, \href{https://doi.org/10.1023/A:1007647118522}{\textit{Lett. Math. Phys.}}
 \textbf{48} (1999), 277--289.

\bibitem{guan2024}
Guan W., Wang S., Rui W., Cheng J., Lax structure and tau function for large
 {BKP} hierarchy, \href{https://arxiv.org/abs/2404.09815}{arXiv:2404.09815}.

\bibitem{hirota2004}
Hirota R., The direct method in soliton theory, \textit{Cambridge Tracts in
 Math.}, Vol. 155, \href{https://doi.org/10.1017/CBO9780511543043}{Cambridge University Press}, Cambridge, 2004.

\bibitem{Jimbo1983infin}
Jimbo M., Miwa T., Solitons and infinite-dimensional {L}ie algebras,
 \href{https://doi.org/10.2977/prims/1195182017}{\textit{Publ. Res. Inst. Math. Sci.}} \textbf{19} (1983), 943--1001.

\bibitem{Kac2018jjm}
Kac V.G., van~de Leur J.W., Equivalence of formulations of the {MKP} hierarchy
 and its polynomial tau-functions, \href{https://doi.org/10.1007/s11537-018-1803-1}{\textit{Jpn.~J.~Math.}} \textbf{13} (2018),
 235--271, \href{https://arxiv.org/abs/1801.02845}{arXiv:1801.02845}.

\bibitem{Kiso1990ptp}
Kiso K., A remark on the commuting flows defined by {L}ax equations,
 \href{https://doi.org/10.1143/PTP.83.1108}{\textit{Progr. Theoret. Phys.}} \textbf{83} (1990), 1108--1114.

\bibitem{Krichever2022lmp}
Krichever I., Zabrodin A., Constrained {T}oda hierarchy and turning points of
 the {R}uijsenaars--{S}chneider model, \href{https://doi.org/10.1007/s11005-022-01519-0}{\textit{Lett. Math. Phys.}} \textbf{112}
 (2022), 23, 26~pages, \href{https://arxiv.org/abs/2109.05240}{arXiv:2109.05240}.

\bibitem{Krichever2023pd}
Krichever I., Zabrodin A., Toda lattice with constraint of type {B},
 \href{https://doi.org/10.1016/j.physd.2023.133827}{\textit{Phys.~D}} \textbf{453} (2023), 133827, 11~pages, \href{https://arxiv.org/abs/2210.12534}{arXiv:2210.12534}.

\bibitem{Kuper1985cmp}
Kupershmidt B.A., Mathematics of dispersive water waves, \href{https://doi.org/10.1007/BF01466593}{\textit{Comm. Math.
 Phys.}} \textbf{99} (1985), 51--73.

\bibitem{Kuper1995cmp}
Kupershmidt B.A., Canonical property of the {M}iura maps between the m{KP} and
 {KP} hierarchies, continuous and discrete, \href{https://doi.org/10.1007/BF02100590}{\textit{Comm. Math. Phys.}}
 \textbf{167} (1995), 351--371.

\bibitem{Liu2010}
Liu S., Cheng Y., He J., The determinant representation of the gauge
 transformation for the discrete {KP} hierarchy, \href{https://doi.org/10.1007/s11425-010-0067-x}{\textit{Sci. China Math.}}
 \textbf{53} (2010), 1195--1206, \href{https://arxiv.org/abs/0904.1868}{arXiv:0904.1868}.

\bibitem{Liu2024}
Liu Y., Yan X., Wang J., Cheng J., Generalized bigraded {T}oda hierarchy,
 \href{https://doi.org/10.1063/5.0221612}{\textit{J.~Math. Phys.}} \textbf{65} (2024), 103506, \href{https://arxiv.org/abs/2405.19952}{arXiv:2405.19952}.

\bibitem{Mikhailov1979}
Mikhailov A.V., Integrability of a two--dimensional generalization of the
 {T}oda chain, \textit{JETP Lett.} \textbf{30} (1979), 414--418.

\bibitem{Mikhailov1981}
Mikhailov A.V., Olshanetsky M.A., Perelomov A.M., Two-dimensional generalized
 {T}oda lattice, \href{https://doi.org/10.1007/BF01209308}{\textit{Comm. Math. Phys.}} \textbf{79} (1981), 473--488.

\bibitem{Prokofev2023tmp}
Prokofev V.V., Zabrodin A.V., Tau-function of the {B}-{T}oda hierarchy,
 \href{https://doi.org/10.1134/S0040577923110041}{\textit{Theoret. Math. Phys.}} \textbf{217} (2023), 1673--1688,
 \href{https://arxiv.org/abs/2303.17467}{arXiv:2303.17467}.

\bibitem{Shaw1997}
Shaw J.-C., Tu M.-H., Miura and auto-{B}\"acklund transformations for the c{KP}
 and cm{KP} hierarchies, \href{https://doi.org/10.1063/1.532164}{\textit{J.~Math. Phys.}} \textbf{38} (1997),
 5756--5773.

\bibitem{Takasaki2018}
Takasaki K., Toda hierarchies and their applications, \href{https://doi.org/10.1088/1751-8121/aabc14}{\textit{J.~Phys.~A}}
 \textbf{51} (2018), 203001, 35~pages, \href{https://arxiv.org/abs/1801.09924}{arXiv:1801.09924}.

\bibitem{Ueno1982}
Ueno K., Takasaki K., Toda lattice hierarchy, in Group Representations and
 Systems of Differential Equations ({T}okyo, 1982), \textit{Adv. Stud. Pure
 Math.}, Vol.~4, \href{https://doi.org/10.2969/aspm/00410001}{North-Holland}, Amsterdam, 1984, 1--95.

\bibitem{van2015}
van~de Leur J.W., Orlov A.Yu., Pfaffian and determinantal tau functions,
 \href{https://doi.org/10.1007/s11005-015-0786-6}{\textit{Lett. Math. Phys.}} \textbf{105} (2015), 1499--1531,
 \href{https://arxiv.org/abs/1404.6076}{arXiv:1404.6076}.

\bibitem{Yang2022}
Yang Y., Cheng J., Bilinear equations in {D}arboux transformations by
 boson-fermion correspondence, \href{https://doi.org/10.1016/j.physd.2022.133198}{\textit{Phys.~D}} \textbf{433} (2022), 133198,
 30~pages, \href{https://arxiv.org/abs/2022.13319}{arXiv:2022.13319}.

\bibitem{zabrodin}
Zabrodin A.V., Hirota difference equations, \href{https://doi.org/10.1007/BF02634165}{\textit{Theoret. and Math. Phys.}}
 \textbf{113} (1997), 1347--1392.

\bibitem{zakharov1}
Zakharov V.E., Mikhailov A.V., Relativistically invariant two-dimensional
 models of field theory which are integrable by means of the inverse
 scattering problem method, \textit{Sov. Phys. JEPT} \textbf{47} (1978),
 1017--1027.

\bibitem{zakharov2}
Zakharov V.E., Takhtadzhyan L.A., Equivalence of a nonlinear {S}chr\"odinger
 equation and a~{H}eisenberg ferromagnet equation, \href{https://doi.org/10.1007/BF01030253}{\textit{Theoret. and Math.
 Phys.}} \textbf{38} (1979), 17--23.

\end{thebibliography}
\end{document}